\documentclass[sigconf,nonacm]{acmart}

\AtBeginDocument{%
  \providecommand\BibTeX{{%
    \normalfont B\kern-0.5em{\scshape i\kern-0.25em b}\kern-0.8em\TeX}}}

\usepackage{tikz}
\usetikzlibrary{arrows}
\usepackage{stmaryrd}
\usepackage{mathtools}
\usepackage{todonotes}
\usepackage{booktabs}
\usepackage{comment}
\newcommand{\dfn}{\coloneqq}
\usepackage{enumerate}

\newcommand{\KM}{\ensuremath{\mathfrak{K}}}
\newcommand{\ddfn}{\Coloneqq}

\newcommand{\traces}{\mathrm{Traces}}

\newcommand{\LL}{\protect\ensuremath{\mathrm{L}}\xspace}
\newcommand{\LTL}{\protect\ensuremath{\mathrm{LTL}}\xspace}
\newcommand{\CTL}{\protect\ensuremath{\mathrm{CTL}}\xspace}
\newcommand{\CTLS}{\protect\ensuremath{\mathrm{CTL}^*}\xspace}
\newcommand{\teamltl}{\protect\ensuremath{\mathrm{Team\LTL}}\xspace}
\newcommand{\exteamltl}{\protect\ensuremath{\exists\teamltl}}
\newcommand{\ateamltl}{\protect\ensuremath{\forall\teamltl}}
\newcommand{\teamctl}{\protect\ensuremath{\mathrm{Team}\CTL}}
\newcommand{\exteamctl}{\protect\ensuremath{\exists\teamctl}}
\newcommand{\ateamctl}{\protect\ensuremath{\forall\teamctl}}
\newcommand{\teamctls}{\protect\ensuremath{\mathrm{Team}\CTLS}}
\newcommand{\exteamctls}{\protect\ensuremath{\exists\teamctls}}
\newcommand{\ateamctls}{\protect\ensuremath{\forall\teamctls}}

\newcommand{\Sub}{\mathrm{Sub}}
\newcommand{\SP}{\mathrm{SP}}
\newcommand{\cl}{\mathrm{cl}}
\newcommand{\tef}{tef\xspace}
\newcommand{\tefs}{tefs\xspace}
\newcommand{\Con}{\mathrm{Con}}

\newcommand{\Um}{\mathsf{U}}

\newcommand{\U}{\mathsf{U}}

\newcommand{\X}{\mathsf{X}}
\newcommand{\F}{\mathsf{F}}
\newcommand{\Gm}{\mathsf{G}}
\newcommand{\W}{\mathsf{W}}
\DeclareMathOperator{\op}{\mathsf{op}}

\newcommand{\ap}{\mathsf{AP}}

\newcommand{\set}[1]{\left\{#1\right\}}
\newcommand{\setprop}[2]{\left\{\,#1\;\middle|\;#2\,\right\}}
\newcommand{\eval}[1]{\left\llbracket #1\right\rrbracket}

\usepackage{xspace}

\newcommand{\bor}{\varovee}
\DeclareMathOperator{\bneg}{\sim}
\newcommand{\dep}{\mathsf{dep}}
\newcommand{\NE}{\mathrm{NE}}

\makeatletter
\newcommand{\oset}[3][0ex]{%
  \mathrel{\mathop{#3}\limits^{
    \vbox to#1{\kern-2\ex@
    \hbox{$\scriptstyle#2$}\vss}}}}
\makeatother

\newcommand{\flatop}{\oset{\scriptscriptstyle1}{\mathsf{A}}}

\newcommand{\fsSAT}{\mathrm{fs\text-SAT}($TE$)}
\newcommand{\fsMC}{\mathrm{fs\text-MC}($TE$)}
\newcommand{\PC}{\mathrm{PC}($TE$)}

\newcommand{\N}{\mathbb{N}}

\newcommand{\phisynch}{\varphi_{\textsf{synch}}}
\newcommand{\phioff}{\varphi_{\textsf{off}}}
\newcommand{\phifair}{\varphi_{\textsf{fair}}}
\newcommand{\psisynch}{\psi_{\textsf{synch}}}
\newcommand{\switch}{\mathrm{switch}}
\makeatletter
\DeclareRobustCommand\bigop[1]{%
  \mathop{\vphantom{\sum}\mathpalette\bigop@{#1}}\slimits@
}
\newcommand{\bigop@}[2]{%
  \vcenter{%
    \sbox\z@{$#1\sum$}%
    \hbox{\resizebox{\ifx#1\displaystyle.9\fi\dimexpr\ht\z@+\dp\z@}{!}{$\m@th#2$}}%
  }%
}
\makeatother
\newcommand{\bigvarovee}{\DOTSB\bigop{\varovee}}

\newcommand{\te}[3]{\tau\ifx\empty#3\else_{#3}\fi\ifx\empty#1\else(#1)\fi\ifx\empty#2\else[#2]\fi}

\makeatletter
\newtheorem*{rep@theorem}{\rep@title}
\newcommand{\newreptheorem}[2]{%
\newenvironment{rep#1}[1]{%
 \def\rep@title{#2 \ref{##1}}%
 \begin{rep@theorem}}%
 {\end{rep@theorem}}}
\makeatother

\newreptheorem{theorem}{Theorem}
\newreptheorem{proposition}{Proposition}

\usepackage{hyperref}

\usepackage{array}
\newcolumntype{C}[1]{>{\centering\let\newline\\\arraybackslash\hspace{0pt}}m{#1}}

\copyrightyear{2022} 
\acmYear{2022} 
\setcopyright{rightsretained} 
\acmConference[LICS '22]{37th Annual ACM/IEEE Symposium on Logic in Computer Science (LICS)}{August 2--5, 2022}{Be'er Sheva, Israel}
\acmBooktitle{37th Annual ACM/IEEE Symposium on Logic in Computer Science (LICS) (LICS '22), August 2--5, 2022, Be'er Sheva, Israel} 
\acmDOI{10.1145/3531130.3533360}
\acmISBN{978-1-4503-9351-5/22/08}

\begin{document}

\title{Temporal Team Semantics Revisited}

\author{Jens Oliver Gutsfeld}
\email{jens.gutsfeld@uni-muenster.de}
\author{Christoph Ohrem}
\email{christoph.ohrem@uni-muenster.de}
\affiliation{%
  \institution{Institut f\"ur Informatik, Universität M\"unster}
  \city{M\"unster}
  \country{Germany}
}

\author{Arne Meier}
\email{meier@thi.uni-hannover.de}
\orcid{0000-0002-8061-5376}
\affiliation{%
  \institution{Institut f\"ur Theoretische Informatik, Leibniz Universit\"at Hannover}
  \city{Hanover}
  \country{Germany}}

\author{Jonni Virtema}
\email{j.t.virtema@sheffield.ac.uk}
\orcid{0000-0002-1582-3718}
\affiliation{%
  \institution{Department of Computer Science, University of Sheffield}
  \city{Sheffield}
  \country{United Kingdom}
}


\begin{abstract}
Hyperproperties are an influential framework with the ability to describe important properties such as information flow policies, security requirements or relations between executions of threads in parallel systems.
Recently, temporal logics have been studied as an approach to the specification of hyperproperties, resulting in the conception of ``hyperlogics''.
With a few recent exceptions, the hyperlogics thus far developed can only relate different traces of a transition system in a synchronous manner, although important information is contained in the relation between different points in their asynchronous interaction.
In order to specify such ``asynchronous hyperproperties'', new trace quantifier based hyperlogics have been developed.
However, certain requirements describing relations between all executions of a system cannot be expressed in hyperlogics with trace quantification.
Additionally, these logics induce model checking problems with prohibitively high model checking complexity costs in the number of quantifier alternations.
In this paper, we study an alternative approach to asynchronous hyperproperties by introducing a novel foundation of temporal team semantics.
Team semantics is a logical framework that specifies properties of sets of traces of unbounded size directly, and thus does not have the same limitation as the quantifier based logics mentioned above.
We consider three logics: TeamLTL, TeamCTL and TeamCTL* which employ quantification over so-called ``time evaluation functions'' controlling the asynchronous progress of traces instead of quantification over traces. 
The use of time evaluation functions constitutes a novel approach to define expressive logics for hyperproperties where diverse asynchronous interactions between computations can be formalised and enforced.
We relate synchronous TeamLTL to our new logics and show how it can be embedded into them.
We show that the model checking problem for exists-TeamCTL with Boolean disjunctions is highly undecidable by encoding recurrent computations of non-deterministic 2-counter machines.
Finally, we present a translation from TeamCTL* to Alternating Asynchronous B\"uchi Automata (AABA), and obtain decidability results for the path checking problem as well as restricted variants of the model checking and satisfiability problems.
For the decidable fragments, the complexity is independent of the quantifier depth and indeed polynomial time in the size of the input system for fixed formulae.
Our translation constitutes the first approach to team semantics based on automata-theoretic methods.

\end{abstract}



\maketitle

\section{Introduction}

Since the 1980s, model checking has become a staple in verification. For Linear Temporal Logic (LTL) and its progeny, the model checking problem asks whether every trace of a given system fulfils a given temporal specification such as a liveness or fairness property. 
Notably, this specification considers the traces of the input in isolation and cannot relate different traces to each other.
However, it is not hard to come up with natural properties that require viewing different traces in tandem. 
For example, asking whether the value of a variable $x$ on average exceeds some constant $c$ amounts to summing up the value of $x$ for \textit{all} traces and then averaging the result. In this context, the information given by a single trace viewed in isolation is of little avail.
Likewise, it does not suffice to consider properties of isolated traces, when we consider executions of parallel programs in which individual threads are represented by single traces.
The same is true for information flow properties of systems like observational determinism or generalised non-interference.
This need to be able to specify properties of collections of traces has led to the introduction of the notion of a \textit{hyperproperty} \cite{Clarkson+Schneider/10/Hyperproperties}.
Technically, a trace property is just a set of traces, and vice versa. Hyperproperties on the other hand describe properties of sets of traces, and thus correspond to sets of sets of traces.
Since established temporal logics like LTL  can express only trace properties, but not genuine hyperproperties, new logics were developed for hyperproperties. 
Generally, the approach has been to pick a temporal logic defined on traces and lift it to sets of traces by adding quantification over named paths, for example, LTL becomes HyperLTL \cite{clarkson14}, QPTL becomes HyperQPTL \cite{Rabe16}, PDL-$\Delta$ becomes HyperPDL-$\Delta$ \cite{GutsfeldMO20} and so on. \looseness=-1

A promising alternative approach for lifting temporal logics to hyperproperties is to shift to the so-called \emph{team semantics}.
In the past decade, team logics have established themselves as a vibrant area of research~\cite{Abramsky2016,grdel_et_al:DR:2019:10568}.
The term \emph{team semantics} was coined by V\"a\"an\"anen~\cite{DBLP:books/daglib/0030191}, inspired by the earlier work of Hodges~\cite{Hodges97c}. 
The idea behind all logics utilising team semantics is to evaluate formulae, not over single states of affairs such as assignments or traces, but over sets of such states of affairs (i.e., over \emph{teams}).
Soon after its origin, team semantics was already applied to first-order, propositional, and modal settings. 
At the heart of these logics lies the ability to enrich the logical language with novel atomic formulae for stating properties of teams. The most prominent of these atoms is the \emph{dependence atom} $\dep(\bar x,\bar y)$ stating that the variables $\bar x$ functionally determine the values of $\bar y$ with respect to some given team (a set of assignments). Another important atom is the \emph{inclusion atom} $\bar x \subseteq \bar y$ expressing the inclusion dependency that all the values that occur for $\bar x$ in a given team, also occur as a value for $\bar y$.
Team logics implement concepts and formalisms from a wealth of different disciplines such as statistics and database theory~\cite{grdel_et_al:DR:2019:10568}. 
While the bulk of the research has concerned itself with logics expressing qualitative properties of data, recent discoveries in multiset \cite{DBLP:journals/tocl/GradelW22} and probabilistic \cite{HannulaKBV20,DBLP:conf/jelia/HannulaV21} variants of team semantics have shifted the focus to the quantitative setting.


Recently, Krebs~et~al.~\cite{kmvz18} made an important advancement in the field by introducing the first team-based temporal logics for hyperproperties. 
The logic $\teamltl$ does not add quantifiers or names to LTL, but instead achieves the lifting to hyperproperties by adapting team semantics, i.e., by evaluating formulae directly over sets of traces and adding new atomic statements that can be used to express hyperproperties such as \emph{non-inference} directly.
Like LTL, but unlike HyperLTL and related logics, $\teamltl$ retains the property of being a purely combinatory, trace quantifier-free logic. 
It can also express specifications for which no analogue in HyperLTL is available.
Most works on the named quantifier approach and the team semantics approach concentrate on synchronous interactions between traces.
In 2018, Krebs~et~al.~\cite{kmvz18} introduced two different semantics for $\teamltl$: a synchronous one and an asynchronous one. 
They can be seen as polar opposites: in the first, computations progress in lockstep, and in the second, there is no way of relating the passage of time between distinct traces.
This asynchronous semantics is rather weak and cannot deal with the plethora of ways asynchronicity occurs in real-world systems. For example, in order to capture multithreaded environments in which processes are not scheduled lockstepwise but still have some rules governing the computation, a setting that can model different modes of asynchronicity is required. 
The ubiquity of asynchronicity thus calls for the development of new hyperlogics that can express asynchronous specifications.
In 2021, Gutsfeld~et~al.~\cite{GutsfeldMO21} conducted the first systematic study of asynchronous hyperproperties and introduced both a temporal fix-point calculus, H$_\mu$, and an automata-theoretic framework, Alternating Asynchronous Parity Automata (AAPA), to tackle this class of properties.

\paragraph{Our contribution.} In this paper, we present a new approach to $\teamltl$ by using explicit quantification over \textit{time evaluation functions} (tefs for short) that describe the asynchronous interleavings of traces. This allows for the fine-grained use of asynchronicity in $\teamltl$ specifications.
Using the new approach, we reconstruct the semantics of $\teamltl$ from scratch, thereby also defining novel team semantics variants of $\CTL$ and $\CTLS$ (i.e., $\teamctl$ and $\teamctls$).
As an example (see Sections \ref{sec:logics} and \ref{sec:extensions} for the precise semantics), let $o_1, \ldots, o_n$ be observable outputs, $c_1, c_2$ be confidential outputs, and $s$ be a secret. The formula $\varphi \dfn \Gm_\forall (o_1,\dots,o_n, s) \subseteq (o_1,\dots, o_n, \neg s)$ expresses a form of \emph{non-inference} by stating that independent of the asynchronous behaviour of the system, the observer cannot infer the current value of the secret from the outputs. The formula $\psi \dfn \Gm_\exists \dep(c_1,c_2,s)$ expresses that for some asynchronous behaviour of the system, the confidential outputs functionally determine the secret. Finally, the formula $\psi \lor \varphi$ states that the executions of the system can be decomposed into two parts; in the first part, the aforementioned dependence holds, while in the second part, the non-inference property holds.

We wish to emphasise that we are not redefining asynchronous $\teamltl$ (or synchronous $\teamltl$ for that matter). Our goal is to define semantics for $\teamltl$ that is versatile enough to deal with the plethora of different modes of asynchronicity that occur in real-world applications. We propose a formalism that can express both synchronous and asynchronous behaviour. Indeed, we show that synchronous $\teamltl$ can be embedded into our new logics. Asynchronous $\teamltl$ (as defined by Krebs~et~al.~\cite{kmvz18}) cannot be directly embedded into our new setting, for each time evaluation function describes a dependence between the global clock and the local clocks. In asynchronous $\teamltl$, no such dependence exists and the setting resembles somewhat our $\teamctl$, albeit with modified semantics.

Besides quantified \tefs and LTL constructs,  we also study several extensions of TeamLTL with different atoms from the team semantics literature, e.g. the dependence and inclusion atoms mentioned above. 
We establish that our logics provide a unifying framework in which previous temporal logics with team semantics can be embedded in an intuitive and efficient manner.
We show that the model checking problem is highly undecidable already for the extension of $\exteamctl$ with the Boolean disjunction.
However, we also present a translation from $\teamctls$ to Alternating Asynchronous B\"uchi Automata (AABA), a subset of AAPA with a B\"uchi condition, over finite teams of fixed size. This translation allows us to transfer \emph{restricted interleaving semantics} for AAPA, such as the \emph{$k$-synchronous} and \emph{$k$-context-bounded} semantics of \cite{GutsfeldMO21}, to $\teamctls$ and employ decidability results for these restricted semantics for the path checking problem and finite variants of the satisfiability and model checking problems.
This translation is of independent interest because it constitutes the first application of automata-theoretic methods in the context of team semantics and dependence logic, and can therefore serve as a cornerstone for further development in this area. 
For the decidable fragments, the complexity of the model checking problem is independent of the quantifier depth of the input formula.
In particular, for fixed formulae, it is polynomial time in the size of the input system, which is not the case for HyperLTL and its progeny. 
Our approach constitutes the first approach to team semantics based on automata-theoretic methods.
Our complexity results for the model checking problem are depicted on page~\pageref{tbl:overview} in Table~\ref{tbl:overview}.

\paragraph{Related work.}
Hyperlogics that add quantifiers for named paths were studied in\cite{clarkson14,finkbeinerRS15,BozzelliMP15,Finkbeiner017,GutsfeldMO20,BeutnerF21}. These logics are orthogonal to ours as they do not involve team semantics.
There are also multiple works on hyperlogics with team semantics
\cite{KontinenS21,DBLP:conf/time/KrebsMV15,kmvz18,DBLP:journals/corr/abs-2010-03311,DBLP:journals/tcs/Luck20}. 
These \textit{team logics} either have purely synchronous semantics or do not allow for fine-grained control of the asynchronicity as do our time evaluation functions and fragments with restricted asynchronous semantics.
The recent work of Kontinen et al. \cite{KontinenSV23} has revealed fascinating connections between asynchronous TeamLTL and HyperLTL, as well as pinpointed computationally well behaved fragments of asynchronous TeamLTL.

In recent years, several works have studied temporal logics for asynchronous hyperproperties \cite{GutsfeldMO21,DBLP:conf/cav/BaumeisterCBFS21,DBLP:conf/lics/BozzelliPS21,Bonakdarpour2020,DBLP:conf/cav/BeutnerF22,DBLP:journals/pacmpl/GutsfeldMO24}.
Gutsfeld~et~al.~\cite{GutsfeldMO21} study asynchronous hyperproperties using the fixed-point calculus H$_{\mu}$ and Alternating Asynchronous Parity Automata. 
Other asynchronous variants of HyperLTL have been introduced in \cite{DBLP:conf/cav/BaumeisterCBFS21,DBLP:conf/lics/BozzelliPS21, Bonakdarpour2020,DBLP:conf/cav/BeutnerF22}. 
Bozelli~et~al.~\cite{DBLP:conf/lics/BozzelliPS21} focus on both HyperLTL variants with special modalities referring to stuttering on paths and contexts describing asynchronous behaviour.
Similar to the stuttering mechanism used in \cite{DBLP:conf/lics/BozzelliPS21}, Beutner and Finkbeiner~\cite{DBLP:conf/cav/BeutnerF22} use an observation mechanism and Gutsfeld~et~al. \cite{DBLP:journals/pacmpl/GutsfeldMO24} use a mumbling mechanism to model asynchronous progress.
Bonakdarpour~et~al.~\cite{Bonakdarpour2020} and Baumeister~et~al.~\cite{DBLP:conf/cav/BaumeisterCBFS21} take a different approach and use quantification over so-called \textit{trajectories} which determine the asynchronous interleaving of traces. 
These trajectories are similar to our \tefs, but they are not studied in the context of team semantics or AABA, and no specific analysis of their properties is presented.
As we show, variants of our logics can express properties such as synchronicity or fairness for \tefs, and there is no way in sight to do this in the logic of Baumeister~et~al.~\cite{DBLP:conf/cav/BaumeisterCBFS21}.
There are two systematic studies on the expressive power of hyperlogics \cite{CoenenFHH19,DBLP:conf/concur/BozzelliPS22}.
While Coenen~et~al.~\cite{CoenenFHH19} mainly compare synchronous hyperlogics, Bozzelli~et~al.~\cite{DBLP:conf/concur/BozzelliPS22} focus on asynchronous hyperlogics and compare some of the logics mentioned above.
However, none of the logics they compare employ team semantics. 

\paragraph{Organisation.}
The paper is organised as follows. 
In Section~\ref{sec:preliminaries}, we introduce the necessary preliminaries and notation. 
Section~\ref{sec:tefs-and-temporal-teams} introduces time evaluation functions and temporal teams. 
In Section~\ref{sec:logics}, we define the logics $\teamltl$, $\teamctl$, and $\teamctls$. 
In Section~\ref{sec:extensions}, we introduce extensions of these logics with dependence and inclusion atoms as well as Boolean disjunction, non-emptiness atoms, and the universal subteam quantifier $\flatop$. 
We continue with a consideration of what one can express in these logics in Section~\ref{subsec:tefproperties}. 
In Section~\ref{sec:synchronous}, we start with a connection of synchronous $\teamltl$ and existential $\teamltl$ on the level of satisfiability as well as to the universal fragment regarding validity. 
Afterwards, we construct translations from synchronous $\teamltl$ into different subfragments of $\teamctls$. 
The next section is devoted to the model checking problem for the existential fragment $\teamctl$ with Boolean disjunction. 
Motivated by the undecidability of this problem, we present a translation from $\teamctls$ to Alternating Asynchronous B\"uchi Automata in Section~\ref{sec:translation}.
Finally, we conclude the paper with a discussion of our results and an outlook on future work.

This paper is an extended version of the article ``Temporal Team Semantics Revisited'' by the same authors, which appeared in the proceedings of LICS 2022. 
In the following, we highlight the additions and changes made to the conference version:
\begin{itemize}
	\item We improved the notation for the semantics from $(T,\tau)$ to incorporate timesteps directly via $(T,\tau,i)$.
	\item The full proof of Theorem~\ref{thm:GA} is now presented and the formulation has been revised. The operator $\flatop$ has been removed from the statement. 
	\item We added an example (Example~\ref{example:extended-atoms}) to explain our considerations when defining extended atoms. 
	\item Direct proofs of Propositions~\ref{prop:teamltl-all-not-union-closed} and \ref{prop:teamltl-all-downward-closed} are presented.
	\item A detailed proof of Theorem~\ref{thm:synchltl-sat-validity} is presented.
	\item We restructured Section~\ref{sec:synchTeamLTL-embedding} to better highlight the connections between the presented arguments and our results.
	\item We improved Figure~\ref{fig:trace-counter} to make the relation of the traces easily visible.
	\item We restructured and improved the presentation as well as added proof details for the results from Section~\ref{sec:translation}.
    \item We extended Theorem \ref{thm:aaba-embed} and Corollary \ref{cor:teamctls-k-synch-k-bounded} to cover Boolean negation as well.
\end{itemize}

\section{Preliminaries}\label{sec:preliminaries}

We assume familiarity with complexity theory \cite{DBLP:books/daglib/0018514} and make use of the classes $\mathrm{PSPACE}$, $\mathrm{EXPSPACE}$, and $\mathrm{P}$.
Also, we deal with different degrees of undecidability, e.g., $\Sigma_1^0$ and $\Sigma^1_1$. 
A thorough introduction in this regard can be found in the textbook of Pippenger~\cite{DBLP:books/daglib/0092426}.

\paragraph*{General Notation.} If $\vec{a}=(a_0,\dots,a_{n-1})$ is an $n$-tuple of elements and $i<n$ a natural number, we set $\vec{a}[i] \dfn a_i$. 
For $n$-tuples $\vec{a},\vec{b}\in\N^n$ ($n\in\N \cup \set{\omega}$), write $\vec{a}\leq \vec{b}$ whenever $\vec{a}[i]\leq\vec{b}[i]$, for each $i< n$; write $\vec{a}<\vec{b}$, if additionally $\vec{a}\neq \vec{b}$.

\paragraph*{Multisets.} Intuitively, a multiset is a generalisation of a set that records the multiplicities of its elements. 
The collections $\{a,a,b\}$ and $\{a,b\}$ are different multisets, while they are identical when interpreted as sets.
Here, we encode multisets as sets by appending unique indices to the elements of the multisets.

Let $I$ be some infinite set of indices such as $\N \cup \N^\omega$. 
A \emph{multiset} is a set $A$ of pairs $(i,v)$, where $i\in I$ is an index value and $v$ is a set element, such that $a[0]\neq b[0]$ for all distinct $a,b\in A$. 
Multisets $A$ and $B$ are the \emph{same multisets} (written $A=B$), if there exists a bijection $f\colon I\rightarrow I$ such that $B=\setprop{\big(f(a[0]), a[1]\big)}{a\in A}$. 
Using this notation, the collection $\{a,a,b\}$ can be written, e.g., as $\{(1,a),(2,a),(42,b)\}$. 
When denoting elements $(i,v)$ of multisets, we often drop the indices and write simply the set element $v$ instead of the pair $(i,v)$. 
The \emph{disjoint union} $A\uplus B$ of multisets $A$ and $B$ is defined as the set union $A'\cup B'$, where $A'= \setprop{\big((i,0),v\big)}{(i,v)\in A}$ and $B'= \setprop{\big((i,1),v\big)}{(i,v)\in B }$.

\paragraph*{Kripke Structures and Traces.}
Fix a set $\ap$ of \emph{atomic propositions}.
A \emph{rooted Kripke structure} is a $4$-tuple $\KM=(W, R, \eta , r)$, where $W$ is a finite non-empty set of states, $R\subseteq W^2$ a left-total relation, $\eta\colon W\rightarrow 2^\ap$ a labelling function, and $r\in W$ an initial state of $W$. 
A \emph{path} $\sigma$ through a Kripke structure $\KM=(W,R,\eta,r)$ is an infinite sequence $\sigma \in W^\omega$ such that $\sigma[0]= r$ and $(\sigma[i], \sigma[i + 1]) \in R$ for every $i \geq 0$.
A \emph{trace} is an infinite sequence from $(2^\ap)^\omega$. We call a trace \emph{ultimately periodic} if it has the form $uv^\omega$, where $u$ and $v$ are finite sequences from $(2^\ap)^*$. This ultimately periodic trace is encoded by the pair $(u,v)$.
The \emph{trace of a path $\sigma$} is defined as $t(\sigma) \dfn \eta(\sigma[0])\eta(\sigma[1])\dots \in (2^\ap)^\omega$. 
A Kripke structure $\KM$ induces a multiset of traces, defined as $\traces(\KM) \dfn \setprop{\big(\sigma, t(\sigma) \big)}{\sigma \text{ is a path through $\KM$}}$.

\paragraph*{Temporal Logics.} Let us start by recalling the syntax of $\CTLS$, $\CTL$, and $\LTL$ from the literature~\cite{DBLP:books/daglib/0007403}.
We adopt, as is common in studies on team logics, the convention that formulae are given in negation normal form. 
We again fix a set $\ap$ of atomic propositions. 
The set of formulae of $\CTLS$ (over $\ap$) is generated by the following grammar:
\[
\varphi \ddfn p \mid \neg p  \mid (\varphi \lor \varphi) \mid (\varphi\land \varphi) \mid \X \varphi \mid [\varphi \Um \varphi] \mid [\varphi \W \varphi] \mid \exists \varphi \mid \forall \varphi, 
\]
where $p \in \ap$ is a proposition symbol, $\X$, $\Um$ and $\W$ are temporal operators, and $\exists$ and $\forall$ are path quantifiers.
$\CTL$ is the syntactic fragment of $\CTLS$, where each temporal operator directly follows a path quantifier (and vice versa). In order to simplify the notation, we write $\X_\forall \varphi$, $\psi\Um_\exists \varphi$, etc., instead of $\forall \X \varphi$ and $\exists [\psi\Um \varphi]$.
That is, the $\CTL$ syntax (over $\ap$) is given by the grammar:
\begin{multline*}
\varphi \ddfn p \mid \neg p  \mid (\varphi \lor \varphi) \mid (\varphi\land \varphi) \mid \X_\exists \varphi \mid \X_\forall \varphi \mid \\
[\varphi \Um_\exists \varphi] \mid [\varphi \Um_\forall \varphi] \mid [\varphi \W_\exists \varphi]\mid [\varphi \W_\forall \varphi], 
\end{multline*}
where $p \in \ap$.
Finally, $\LTL$ is the syntactic fragment of $\CTLS$ without any path quantifiers.
The Kripke semantics for $\CTLS$ is defined in the usual manner with respect to Kripke structures and traces generated from them~\cite{DBLP:reference/mc/PitermanP18}.
For an $\LTL$-formula $\varphi$, a trace $t$, and $i\in\N$, we write $\eval{\varphi}_{(t,i)}$ for the truth value of $(t,i)\Vdash\varphi$ using standard LTL Kripke semantics.
The logical constants $\top,\bot$ and connectives $\rightarrow,\leftrightarrow$ are defined as usual (e.g., $\bot \dfn (p \land \neg p)$).
Furthermore, we use modalities $\F \varphi\coloneqq[\top\Um\varphi]$ and $\Gm\varphi\coloneqq[\varphi\W\bot]$.

\section{Revisiting temporal team semantics}

In this section, we return to the drawing board and reconstruct the semantics of $\teamltl$ from scratch. By doing so, we end up also defining team semantics variants of $\CTL$ and $\CTLS$ (i.e., $\teamctl$ and $\teamctls$). Our starting goal is to consider hyperproperties in a setting where synchronicity of the passage of time between distinct computation traces is not presupposed. Instead, we stipulate a global lapse of time (global clock) and relate the lapse of time on computation traces (local clocks) to the lapse of time on the global clock using a concept we call \emph{time evaluation functions}. 
Our approach here is similar to the one of Baumeister~et~al.~\cite{DBLP:conf/cav/BaumeisterCBFS21} and Bonakdarpour ~et~al.~ \cite{Bonakdarpour2020}, where our \emph{time evaluation functions} are called \emph{trajectories}. 
\begin{table}[t]
	\centering
\begin{tabular}{ll}\toprule
	\textbf{Property of tef} & \textbf{Definition}\\\midrule
	Monotonicity& $\forall i\in \N:\tau(i)\leq \tau(i+1)$\\
	Strict Mon.& $\forall i\in \N:\tau(i)< \tau(i+1)$\\
	Stepwiseness& $\forall i\in \N: \tau(i)\leq \tau(i+1) \leq \tau(i)+\vec 1$\\
	*Fairness& $\forall i\in \N\, \forall t\in T\, \exists j\in\N: \tau(j,t)\ge i$\\
	*Non-Parallelism& $\forall i\in \N: i=\sum_{t\in T}\tau(i,t)$\\
	*Synchronicity& $\forall i\in\N\,\forall t,t'\in T:\tau(i,t)=\tau(i,t')$\\
	\bottomrule
\end{tabular}
	
	\caption{Some properties of \tefs. \textbf{*} marks optional properties. 
 }\label{tbl:tef-props}
\end{table}
\subsection{Time evaluation functions and temporal teams}\label{sec:tefs-and-temporal-teams}

Given a (possibly infinite) multiset of traces $T$, a \emph{time evaluation function} (\emph{\tef} for short) for $T$ is a function $\tau\colon \N\times T \rightarrow \N$ that, given a trace $t\in T$ and a value of the global clock $i\in \N$, outputs the value $\tau(i,t)$ of the local clock of trace $t$ at global time $i$.
We write $\tau(i)$ to denote the tuple $\big(\tau(i,t)\big)_{t\in T}$ obtained using some canonical order on $T$. 
Needless to say, not all functions $\tau\colon \N\times T \rightarrow \N$ satisfy properties that a function should \emph{a priori} satisfy in order to be called a time evaluation function. 
We refer the reader to Table \ref{tbl:tef-props} for a list of \tef properties considered in this paper. 
Intuitively, a tef is \emph{monotonic} if the values of the local clocks only increase;
\emph{strict monotonicity} requires that at least one local clock  advance in every step; 
\emph{stepwiseness} refers to local clocks advancing at most one step each time;
\emph{fairness} implies that no local clock gets stuck infinitely long;
\emph{non-parallelism} forces exactly one clock to advance each time step;
\emph{synchronicity} means that all clocks advance in lockstep.

In the current paper, we are designing logics for hyperproperties of discrete linear time execution traces. It is thus clear that all \tefs should at least satisfy \emph{monotonicity}. One design principle of our setting is to use a global reference clock in addition to the local clocks of the computations. As the global clock can be seen to specify the granularity of the passage of time, it is natural to assume that no local clock can advance faster than the global clock. Thus, we stipulate that every \tef must satisfy \emph{stepwiseness}.
Finally, a crucial property that should hold for all logics with team semantics is that $\teamltl$ should be a conservative extension of LTL. 
That is, on singleton teams, $\teamltl$ semantics should coincide with the semantics of standard non-team-based LTL. 
In order, for example, for the next operator $\X$ to enjoy this invariance between LTL and $\teamltl$, we stipulate \emph{strict monotonicity} instead of simple \emph{monotonicity} as a property for all \tefs.

Table \ref{tbl:tef-props} also lists some optional properties of \tefs.
\emph{Synchronicity} links the passings of time between different computation traces.
When this property is assumed, one obtains synchronous TeamLTL.
Synchronicity is incompatible with \emph{non-parallelism} and, in conjunction with strict monotonicity, implies \emph{fairness}.
The former is a desirable property in some contexts and undesirable in others.
Our choice of making \emph{fairness} optional might feel controversial to some.
However, we argue that having the possibility that some computation traces may freeze is more natural than disallowing it completely.
Moreover, in Section \ref{subsec:tefproperties} we establish that in many cases our logics are expressive enough to specify that a \tef must satisfy fairness.

We arrive at the following formal definitions. For technical reasons, we will define our semantics by using stuttering \tefs (defined below). Note that stuttering \tefs can arise when \tefs are restricted to some subset $T'\subseteq T$ of traces (see the definition of disjunction in Section \ref{sec:logics}).
\begin{definition}
A function of the type $\N\times T \rightarrow \N$ is
a \emph{stuttering \tef for $T$} if it satisfies monotonicity and stepwiseness,
a \emph{\tef for $T$} if it satisfies strict monotonicity and stepwiseness, and
a \emph{synchronous \tef for $T$} if it satisfies strict monotonicity, stepwiseness, and synchronicity.
A \tef is \emph{initial}, if $\tau(0,t)=0$ for each $t\in T$.
\end{definition}

\begin{definition}
A \emph{temporal team} is a pair $(T,\tau)$, where $T$ is a multiset of traces and $\tau$ is a \tef for $T$.
A pair $(T,\tau)$ is called a \emph{stuttering temporal team} if $\tau$ is a stuttering \tef for $T$.
\end{definition}

\subsection{\texorpdfstring{\boldmath}{}TeamLTL, TeamCTL, and TeamCTL\texorpdfstring{$^*$\unboldmath}{*}}\label{sec:logics}

Let $(T, \tau)$ be a stuttering temporal team.
Team semantics for $\LTL$ (i.e., $\teamltl$) is defined recursively as follows with respect to every $i\in\N$. 
{\small
\begin{align*}
	&(T, \tau,i )\models p   &&\text{iff}&& \forall t \in T: p\in t[\tau(i,t)]\\
	&(T, \tau,i )\models \lnot p  &&\text{iff}&& \forall t \in T: p\notin t[\tau(i,t)]\\
	&(T, \tau,i )\models (\varphi\land\psi)  &&\text{iff}&& (T, \tau, i)\models\varphi  \text{ and } (T, \tau, i)\models\psi\\
	&(T, \tau, i)\models (\varphi\lor\psi)  &&\text{iff}&& \exists T_1\uplus T_2=T:(T_1, \tau, i)\models\varphi \text{ and }(T_2, \tau, i)\models\psi\\
	&(T, \tau, i) \models\X\varphi  &&\text{iff}&& (T, \tau, i+1)\models\varphi\\
	&(T, \tau, i)\models[\varphi\Um\psi] &&\text{iff}&& \exists k\geq i\text{ such that } (T, \tau, k)\models\psi \text{ and }\\
	&&&&&\quad\quad \forall m: i \leq m <k \Rightarrow (T, \tau, m)\models\varphi\\
	&(T, \tau, i)\models[\varphi\W\psi] &&\text{iff}&&\forall k\geq i : (T, \tau, k)\models\varphi \text{ or }\\
	&&&&&\quad\exists m \text{ s.t. } 
	i\leq m\leq k \text{ and } (T, \tau, m)\models\psi 
\end{align*}
}
Note that $(T,\tau,i)\models \bot$ iff $T=\emptyset$, where $\bot \dfn (p\lor \neg p)$.
One reason for setting that $(\emptyset,\tau,i)\models \bot$ is for obtaining equivalencies between formulae of the form $\varphi$ and $(\varphi\lor \bot)$.
If $\tau$ is an initial synchronous \tef, we obtain the synchronous team semantics of $\LTL$ as defined by Krebs~et~al.~\cite{kmvz18}.

In the literature, there exist two variants of team semantics for the split operator $\lor$: strict and lax semantics \cite{galliani12}. 
Strict semantics enforces the split to be a partition whereas lax does not.
Very recent work by Barlag~et~al.~\cite{BarlagHKPV23} gives insight into the reasons behind the different definitions.
The general definition of disjunction in team semantics is based on convex combinations; a team $T$ satisfies a disjunctive formula $(\varphi_1 \lor \varphi_2)$, if $T$ can be constructed as a convex combination of teams $T_1$ and $T_2$ such that $T_i\models \varphi_i$, for $i\in \{1,2\}$.
The concrete interpretation of $\lor$ then depends on how the notion of convex combination is defined in a particular setting. In set-based semantics (where team elements can be seen to be annotated with elements of the Boolean semiring; see \cite{BarlagHKPV23} for further details) disjunction will be interpreted as a set union, while in multiset-based semantics (where team elements can be seen to be annotated with elements of the semiring of natural numbers) the correct interpretation of disjunction is a disjoint union.

While any given multiset of traces $T$ induces a unique initial synchronous \tef, the same does not hold for \tefs in general.
Consequently, two different modes of $\teamltl$ satisfaction naturally emerge: \emph{existential} (a formula is satisfied by some initial \tef) and \emph{universal} (a formula is satisfied by all initial \tefs) satisfaction. In the special case where a unique initial \tef exists, these two modes naturally coincide.

Given a multiset of traces $T$ and a formula $\varphi \in \teamltl$, we write $T\models_\exists \varphi$ if $(T,\tau,0)\models \varphi$ for some initial \tef $\tau$ of $T$. Likewise, we write $T\models_\forall \varphi$ if $(T,\tau,0)\models \varphi$ for all initial \tefs $\tau$ of $T$.
Finally, we write  $T\models_s \varphi$ if $(T,\tau,0)\models \varphi$ for the unique initial synchronous \tef $\tau$ of $T$.

We sometimes refer to the universal and existential interpretations of satisfaction by using $\ateamltl$ and $\exteamltl$, respectively. 
For referring to the synchronous interpretation, we write \emph{synchronous $\teamltl$}.


$\teamctl$ and $\teamctls$ loan their syntax from $\CTL$ and $\CTLS$, respectively. However, while the quantifiers $\exists$ and $\forall$ refer to path quantification in $\CTL$ and $\CTLS$, in the team semantics setting the quantifiers to range over \tefs.
The formal semantics of the quantifiers are as one would assume:
{\begin{align*}
(T, \tau, i) \models \exists \varphi \text{ iff }& (T, \tau', i) \models \varphi \text{ for some \tef $\tau'$ of $T$ s.t.}\\ 
&\tau'(j,t)=\tau(j,t)\text{ for all }j\leq i \text{ and } t\in T,\\
(T, \tau, i) \models \forall \varphi \text{ iff }& (T, \tau', i) \models \varphi \text{ for all \tefs $\tau'$ of $T$ s.t.}\\
&\tau'(j,t)=\tau(j,t)\text{ for all }j\leq i \text{ and } t\in T,
\end{align*}}
We write $\exteamctl$ and $\ateamctl$ to denote the fragments of $\teamctl$ without the modalities $\{ \Um_\forall,  \W_\forall, \X_\forall\}$ and $\{\Um_\exists, \W_\exists, \X_\exists\}$, respectively. Likewise, we write  $\exteamctls$ and $\ateamctls$ to denote the fragments of $\teamctls$ without the quantifier $\forall$ and $\exists$, respectively. We extend the notation $\models_\exists$ and $\models_\forall$ to $\teamctls$-formulae as well.

In this paper, we consider the following decision problems for different combinations of logics $L \in \{\teamltl, \teamctl,$ $\teamctls\}$ and modes of satisfaction $\models_{*} \in \{\models_\exists,\models_\forall,\models_s\}$.
\begin{description}
	\item[Satisfiability:] Given an $(L,\models_{*})$-formula $\varphi$, is there a multiset of traces $T$ such that $T \models_{*} \varphi$?
	\item[Model Checking:] Given an $(L,\models_{*})$-formula $\varphi$ and a Kripke structure $\KM$, does $\traces(\KM) \models_{*} \varphi$ hold?
	\item[Path Checking:] Given an $(L,\models_{*})$-formula $\varphi$ and a finite multiset of ultimately periodic traces $T$, does $T \models_{*} \varphi$?
\end{description}

\subsection{\texorpdfstring{\boldmath}{}Extensions of TeamLTL, TeamCTL, and TeamCTL\texorpdfstring{$^*$\unboldmath}{*}}\label{sec:extensions}
Team logics can easily be extended by atoms describing properties of teams.
These extensions are a well-defined way to delineate the expressivity and complexity of the logics we consider.
The most studied of these atoms are {\em dependence atoms} $\dep(\varphi_1,\dots,\varphi_n,\psi)$ and {\em inclusion atoms} $\varphi_1,\dots,\varphi_n  \subseteq  \psi_1,\dots,\psi_n$, where $\varphi_1,\dots,\varphi_n ,\psi,\psi_1,\dots,\psi_n$ are propositional formulae.

In the team semantics literature, atoms whose parameters are propositional variables are often called (proper) atoms, while extended atoms allow arbitrary formulae without atoms in their place. 
Krebs~et~al.~\cite{kmvz18} allowed arbitrary $\LTL$-formulae as parameters. 
Here, we take a middle ground and restrict parameters of atoms to propositional formulae. 
One reason for this restriction is that the combination of extended atoms and time evaluation functions can have unwanted consequences in the $\teamltl$ setting. 
Note that there are two possibilities to define extended atoms. 
We illustrate this with an example.
\begin{example}\label{example:extended-atoms}
        Assume we want to express that ``the value of $\F p$ is constant''.
	Consider a team $T$ and a tef $\tau$.
	The extended atom $\dep(\F p)$ evaluated on $(T,\tau,i)$ could be formally defined in two ways.
        We could interpret $\dep(\F p)$ such that $\F p$ needs to get the same truth value either for all $(t,\tau(i,t))$, where $t\in T$, or for all $(t',i)$, when we define $t'$ to be the trace such that $t'[j] = t[\tau(j,t)]$ for each $j\in\N$.
	In the former, the time-evaluation function $\tau$ is only used for obtaining the state of the local clocks where $\F p$ is to be evaluated.
    In the latter, we take the trace obtained from $t$ by applying the stutterings obtained from $\tau$.
\end{example}
Neither of the above definitions is satisfactory.
In the first interpretation of $\dep(\F p)$, the satisfaction of $(T,\tau,i)\models \dep(\F p)$ is independent of $\tau(j)$, for $j > i$.
That is, the satisfaction of a future modality does not depend on the relation between the global and local clocks in the future.
For the second interpretation, where the evaluation depends on $\tau$, $\dep(\F p)$ does not describe the constancy of $\F p$ with respect to the traces in $T$.

Dependence atoms state that the truth value of $\psi$ is functionally determined by the truth values of all $\varphi_1, \ldots, \varphi_n$. 
Inclusion atoms state that each value combination of $\varphi_1, \ldots, \varphi_n$ must also occur as a value combination of $\psi_1, \ldots, \psi_n$. 
Their formal semantics is defined as follows:
\begin{align*}
	&(T, \tau, i)\models  \dep(\varphi_1,\dots,\varphi_n,\psi)  ~~\text{iff}~~ \forall t ,t' \in T :\\
	&\qquad\bigwedge_{1\leq j\leq n}\eval{\varphi_j}_{(t,\tau(i,t))} = \eval{\varphi_j}_{(t',\tau(i,t'))} \\
	&\qquad\qquad\qquad\qquad\text{ implies } \eval{\psi}_{(t,\tau(i,t))} = \eval{\psi}_{(t',\tau(i,t'))},\\
	&(T, \tau, i)\models \varphi_1,\dots,\varphi_n  \subseteq  \psi_1,\dots,\psi_n  ~~\text{iff}~~ \forall t\in T \, \exists t'\in T :\\
	&\qquad\bigwedge_{1\leq j\leq n}  \eval{\varphi_j}_{(t,\tau(i,t))} = \eval{\psi_j}_{(t',\tau(i,t'))}.
\end{align*}
Notice that, when dependence atoms are used with only one parameter, as in $\dep(\psi)$, they correspond to so-called \emph{constancy atoms}, stating that the truth value $\eval{\psi}$ of the parameter is constant.
We also consider other connectives studied in the team semantics literature:
 \emph{Boolean disjunction} $\bor$, \emph{Boolean negation} $\bneg$, the
	\emph{non-emptiness atom} $\NE$, and the {\em universal subteam quantifier} $\flatop$, with their semantics defined as: 
	\[
	\begin{array}{lcl}
		(T, \tau, i) \models (\varphi \bor \psi)  &\text{iff}& (T,\tau, i) \models \varphi \text{ or } (T,\tau, i) \models \psi\\
		(T, \tau, i) \models \bneg\varphi &\text{iff}& (T,\tau, i) \not\models \varphi \\
		 (T, \tau, i) \models \NE  &\text{iff}& T\neq\emptyset \\
		(T, \tau, i) \models \flatop \varphi  &\text{iff}& \forall t\in T: (\{t\}, \tau, i) \models \varphi
	\end{array}
	\]
If $\mathcal C$ is a collection of atoms and connectives, we denote by $\teamltl(\mathcal C)$  the extension of $\teamltl$ with the atoms and connectives in $\mathcal C$. 
For any atom or connective $\circ$, we write $\teamltl(\mathcal C, \circ)$ instead of $\teamltl(\mathcal C \cup \{\circ\})$.
Most of these extensions will be considered throughout the whole paper while Boolean negation will only be used in Section \ref{sec:translation}.  
A more extensive investigation on the interplay of $\bneg$ within team semantics is given by Lück~\cite{DBLP:journals/lmcs/Luck19}.

It is known that, in the setting of synchronous $\teamltl$, all (all downward closed, resp.) \emph{Boolean properties} of teams are expressible in $\teamltl($$\bor$$, \NE, \flatop)$ (in $\teamltl($$\bor$$,\flatop)$, resp.)~\cite{DBLP:journals/corr/abs-2010-03311}.
Let $B$ be a set of $n$-ary Boolean relations and $\varphi_1,\dots,\varphi_n$ propositional formulae. 
We define the semantics of an expression $[\varphi_1,\dots,\varphi_n]_B$ as follows:
\begin{multline*}
(T, \tau, i) \models [\varphi_1,\dots,\varphi_n]_B \quad\text{iff}\\
\{\, (\eval{\varphi_1}_{(t,\tau(t,i))}, \dots,  \eval{\varphi_n}_{(t,\tau(t,i))} ) \mid t\in T  \,\} \in B.	
\end{multline*}
Expressions of the form $[\varphi_1,\dots,\varphi_n]_B$ are called \emph{generalised atoms}. 
If $B$ is downward closed (i.e., $S\in B$ whenever $S\subseteq R\in B$), it is \emph{a downward closed generalised atom}. 
Dependence and inclusion atoms can also be defined as generalised atoms.

The following was proved in the setting of synchronous $\teamltl$. 
It is, however, easy to check that the same proof works also in our more general setting.
\begin{theorem}[\cite{DBLP:journals/corr/abs-2010-03311}]\label{thm:GA}
Any generalised atom is expressible in $\teamltl(\bor, \NE)$ and all downward closed generalised atoms can be expressed in $\teamltl(\bor)$.	
\end{theorem}
\begin{proof}[Proof sketch]
For any $n$-tuple of Booleans $\vec{b}=(b_1,\dots,b_n)\in \{0,1\}^n$ and $n$-tuple of propositional formulae $\vec{\varphi}=(\varphi_1,\dots,\varphi_n)$, define
\[
\vec{\varphi}^{\vec{b}} \dfn \varphi_1^{b_1}\land \dots \land \varphi_n^{b_n},
\]
where $\varphi^{1} \dfn \varphi$ and $\varphi^{0}$ is defined to be the negation normal form formula equivalent with $\neg \varphi$ obtained via De Morgan's laws. It is now straightforward to check that if $B$ is a downward closed set of $n$-ary Boolean relations and $\vec{\varphi}=(\varphi_1,\dots,\varphi_n)$ in an n-tuple of propositional formulae then
\[
(T, \tau, i) \models [\varphi_1,\dots,\varphi_n]_B \quad\text{ iff }\quad (T, \tau, i) \models \bigvarovee_{R\in B} \bigvee_{\vec{b}\in R} \vec{\varphi}^{\vec{b}}.
\]
If $B$ is not assumed to be downward closed, then
\[
(T, \tau, i) \models [\varphi_1,\dots,\varphi_n]_B \quad\text{ iff }\quad (T, \tau, i) \models \bigvarovee_{R\in B} \bigvee_{\vec{b}\in R} (\vec{\varphi}^{\vec{b}}\land \NE).
\]
\end{proof}

\subsection{Expressing properties of \tefs}\label{subsec:tefproperties}
Recall that we deemed some of the properties of \tefs in Table \ref{tbl:tef-props} optional.
This means that these properties are not required for a function to be considered a \tef but instead constitute subclasses of \tefs.
It is a natural question to ask whether \teamltl can express different classes of hyperproperties when considering \tefs with these optional properties.
Here, we show that this is not the case for fairness and synchronicity in some extensions of \teamltl, since these properties of \tefs are definable in these extensions.

Since fairness and synchronicity make assumptions about the progress of a \tef on each trace, we need a way to track this progress for defining these properties.
For this purpose, we introduce a fresh atomic proposition $o$ that is set on exactly every other position on every trace.
We can then use parity induced by $o$ to enforce synchronicity using the fact that each \tef progresses by exactly one time step whenever the valuation flips from $o$ to $\lnot o$ or vice versa.

\paragraph*{Expressing the alternation on $o$.}
In the context of the model checking problem, ensuring the alternation on $o$ is straightforward.
Given a Kripke structure, create two copies of every state of the structure, one labelled with $o$ and another not labelled with $o$.
Then, transitions from each $o$ labelled state take to the copy of the target state not labelled $o$, and transitions from each state not labelled $o$ take to the copy of the target state labelled $o$.
In this new structure, every trace has the property that $o$ holds on exactly every other index. Moreover, after dropping $o$, the two structures are indistinguishable with respect to their traces.
For the satisfiability problem, the valuation of $o$ is not restricted by any structure.
Thus, we construct a formula that is unsatisfiable by a set of traces that does not correctly alternate on $o$.

\paragraph*{Expressing synchronicity.}

In order to express the synchronicity of a \tef, we encode in a formula that $o$ alternates on all traces simultaneously.
That is, for every step, either $o$ holds on all traces and $\lnot o$ holds in the next step, or vice versa:
\[
	\phisynch \dfn \Gm ((o\land\X \lnot o)\bor (\lnot o\land\X o)).
\]
Note that this formula is unsatisfiable by a set of traces violating the alternation property since all subformulae refer to the whole set of traces.
This formula shows how to encode synchronicity using the Boolean disjunction $\bor$.
We can alter the formula a little to also show that synchronicity can be expressed without this extension:
\[
	\phisynch' \dfn o \land \Gm((o\land\X \lnot o)\lor (\lnot o\land\X o)).
\]
In this formula, we make use of a split $\lor$ instead of the Boolean disjunction $\bor$.
However, since we demand that $o$ hold in the first step, the split can only ever be made true by splitting a set of traces $T$ into $T$ and $\emptyset$.
Thus, $\lor$ behaves like $\bor$ in this formula and can replace the undesired connective.
In Section \ref{sec:synchronous}, we make use of this formula to show how synchronous $\teamltl$ can be embedded into different fragments of $\teamctls$.

\paragraph*{Expressing fairness.}

Fairness can be expressed using the universal subteam quantifier $\flatop$.
Our formula states that for every trace, the valuation of $o$ flips infinitely often, and thus the current \tef never stops making progress on this trace.
This is equivalent to the definition of fairness that requires every index on every trace to be reached.
The formula is:
\[
	\phifair \dfn \flatop \Gm((o \land \F \lnot o) \lor (\lnot o \land \F o)).
\]
Note that for this property, we do not need to enforce strict alternation on $o$. 
We only require that the valuation of $o$ alternates after some finite amount of steps.

\paragraph*{Quantifying \tefs with expressible properties.}
The formulae expressing \tef properties can be used to quantify over \tefs with these properties.
For example, quantification over fair \tefs could be implemented in the following way:
if $T$ is a multiset of traces satisfying alternation for $o$, we have that $(T, \tau, i)\models \exists (\phifair \land \varphi)$ iff there exists a fair tef $\tau'$ of $T$ such that $\tau'(j)=\tau(j)$ for all $j\leq i$ and $(T, \tau', i)\models \varphi$.

\subsection{Basic properties of the logics}

$\teamltl$ is a conservative extension of $\LTL$, i.e., their semantics agree on singleton trace sets.
The next proposition follows by a straightforward inductive argument that is almost identical to the corresponding proof for synchronous $\teamltl$ \cite{kmvz18}. The main insight is that the only initial \tef for a singleton team is the synchronous \tef.
\begin{proposition}
For any $\teamltl$-formula $\varphi$, trace $t$, and initial \tef $\tau$ for $\{t\}$, the following holds:
	$(\{t\}, \tau, 0)\models \varphi \text{ iff } (t,0)\Vdash \varphi$.  
\end{proposition}
Note that in the setting of the above proposition $\teamctl$ and $\teamctls$ both collapse to $\LTL$ as well.

Let $\LL$ be a logic and $\models_*\in \{\models_\exists, \models_\forall, \models_s\}$. We say that $(\LL, \models_*)$ is \emph{downward closed} if, for every $\varphi\in \LL$, $T\models_* \varphi$ implies $S\models_* \varphi$ whenever $S\subseteq T$. Likewise, we say that $(\LL, \models_*)$ is \emph{union closed} if, for every $\varphi\in \LL$, $T\uplus S\models_* \varphi$ holds whenever $T\models_* \varphi$ and $S\models_* \varphi$ hold.

It is known \cite{kmvz18} that $(\teamltl, \models_s)$ is not union closed, but satisfies the downward closure property.
For $\forall$TeamLTL, the former property is established  by Proposition~\ref{prop:teamltl-all-not-union-closed} and the latter by Proposition~\ref{prop:teamltl-all-downward-closed} below. 

\begin{proposition}\label{prop:teamltl-all-not-union-closed}
	$(\teamltl, \models_\forall)$ is not union closed.
\end{proposition}
\begin{proof}
	We define a formula $\phi$ and two teams $T_1,T_2$ such that $T_i\models_\forall \phi$, but $T_1\uplus  T_2\not\models_\forall \phi$. 
%
 We let $t=\{p\}\{p\}\emptyset^\omega$, and define $T_1\dfn\{(1,t)\}$ and $T_2\dfn\{(2,t)\}$. 
	It is easy to check that $T_1\models_\forall \X\X \lnot p$, $T_2\models_\forall \X\X \lnot p$ but $T_1 \uplus T_2\not\models_\forall \X\X \lnot p$.	
 The reason behind this is that in \tefs at least one of its traces advances each step of the global clock (see Table~\ref{tbl:tef-props}). 
	In team $T_1$ and $T_2$ this would yield leaving the $p$-labelled prefix while for $T_1 \uplus T_2$ in the first step one trace can advance and in the next the other one can.
\end{proof}

\begin{proposition}\label{prop:teamltl-all-downward-closed}
	$(\teamltl, \models_\forall)$ is downward closed.
\end{proposition}
\begin{proof}
	Let $S\subseteq T$ be multisets of traces, $\tau$ be some stuttering \tef for $S$, and $\tau'$ a stuttering \tef for $T$, whose reduct to $S$ is $\tau$. We prove by induction on the structure of formulae that 
	\begin{equation}\label{eq:dc}
	(S, \tau, i)\models \varphi \text{ if } (T,\tau',i)\models \varphi,
	\end{equation}
	for every $\teamltl$-formula $\varphi$ and every $i\in\N$. 
	The cases for atomic formulae and conjunction are trivial. 
	For disjunction $\lor$, it suffices to note that splits of $(S,\tau,i)$ can be always copied from the splits of $(T,\tau',i)$ such that the induction hypothesis can be applied to the respective parts of the splits.
	
	Assume then $(T,\tau',i)\models \X \psi$. 
	Then, by the semantics of $\X$, $(T,\tau',i+1)\models \psi$. 
	Since $\tau$ is the $S$ reduct of $\tau'$, by the induction hypothesis, $(S, \tau, i+1) \models \psi$, from which $(S,\tau,i)\models \X \psi$ follows. 
	The cases for $\U$ and $\W$ are similar.
	
	We are now ready to prove the claim of the proposition. 
	Let $S$ and $T$ be as above, and let $\varphi$ be a $\teamltl$-formula such that $T\models_\forall \varphi$. 
	It suffices to show that $(S, \tau, i)\models \varphi$, for every initial \tef $\tau$ of $S$ and $i\in\N$.
	Let $\tau$ be an arbitrary initial \tef for $S$. 
	Clearly, there exists an initial \tef $\tau'$ of $T$, whose $S$ reduct is $\tau$. 
	Now $(T,\tau',i)\models \varphi$, for $T\models_\forall \varphi$ and $\tau'$ is an initial \tef of $T$. 
	Now $(S,\tau,i)\models \varphi$ follows from \eqref{eq:dc}.
\end{proof}

However, unlike synchronous $\teamltl$, the following proposition shows that $(\teamltl, \models_\exists)$ is not downward closed.
\begin{proposition}\label{prop:not-ex}
	$(\teamltl, \models_\exists)$ is not downward closed.
\end{proposition}
\begin{proof}
    Similar to the counterexample used in the proof of Proposition \ref{prop:teamltl-all-not-union-closed}, let $t=\{p\}\{p\}\emptyset^\omega$ be a given trace, and define two teams $T=\{(1,t), (2,t)\}$ and $S=\{(1,t)\}$ that are multisets of traces. 
	It is easy to check that $T\models_\exists \X\X p$, but $S\not\models_\exists \X\X p$.
\end{proof}
The example in the proof of the proposition above also illustrates that the use of multisets of traces is essential in our logics; otherwise, we would violate \emph{locality}. The locality principle dictates that the satisfaction of a formula with respect to a team should not depend on the truth values of proposition symbols that do not occur in the formula.
For this, consider a variant of this example with $t_1=\set{p,q}\set{p}\emptyset^\omega$, $t_2=\set{p}\set{p}\emptyset^\omega$ and $T=\set{t_1,t_2}$. Then, removing $q$ from the traces would yield the set $S=\set{t_2}$ under non-multiset semantics and thus change the truth value of $\X\X p$ as seen in the proof of Proposition \ref{prop:not-ex}. This shows that the use of multiset semantics is vital.
The fact that multiset semantics can be used to retain locality was observed in \cite{DurandHKMV18}.

\section{Fragments of \NoCaseChange{TeamCTL}\texorpdfstring{$^*$}{*} and synchronous \NoCaseChange{TeamLTL}}\label{sec:synchronous}

In this section, we examine connections between our new temporal team logics and the older synchronous $\teamltl$.

\subsection{\texorpdfstring{\boldmath}{}Satisfiability of \texorpdfstring{$\exists$}{exists}TeamLTL and validity of \texorpdfstring{$\forall$}{forall}TeamLTL\texorpdfstring{\unboldmath}{}}

It is straightforward to check that the satisfiability problem of $\exteamltl$ and the validity problem of $\ateamltl$ are, in fact, equivalent to the corresponding problems of synchronous $\teamltl$. 
\begin{theorem}\label{thm:synchltl-sat-validity}
Any given $\teamltl$-formula is satisfiable in $\exteamltl$ if and only if it is satisfiable in synchronous $\teamltl$. Likewise, a given $\teamltl$-formula is valid in $\ateamltl$ if and only if it is valid in synchronous $\teamltl$.
\end{theorem}
\begin{proof}
	Let $\varphi$ be a \teamltl formula.
	First, assume that $\varphi$ is satisfiable in $\exists\teamltl$, i.e., that there is a multiset of traces $T$ such that $T \models_{\exists} \varphi$.
	This means $(T, \tau, 0) \models \varphi$ for some \tef $\tau$.
	Now consider the multiset of traces $T_\tau$ containing for each trace $t \in T$ a trace $t_\tau$ with $t_\tau[i] = t[\tau(i,t)]$ for every $i \in \N$.
	Then, $(T_\tau,\tau_s,0) \models \varphi$ for the synchronous \tef $\tau_s$, which implies that $T_\tau \models_s \varphi$.
	Thus, $\varphi$ is satisfiable in synchronous \teamltl.
	
	Next, assume that $\varphi$ is satisfiable in synchronous \teamltl, i.e., that there is a multiset of traces $T$ such that $T \models_s \varphi$.
	This means that $(T,\tau_s,0) \models \varphi$ for the synchronous \tef $\tau_s$ which in particular implies that $T \models_{\exists} \varphi$ since the synchronous \tef $\tau_s$ is a \tef itself.
	Thus, $\varphi$ is satisfiable in $\exists\teamltl$.
	
	The first argument can also be applied to show that the assumption that $\varphi$ is not valid in $\forall\teamltl$, i.e., there is a multiset of traces $T$ and a tef $\tau$ such that $(T,\tau,0) \not\models \varphi$, implies that $\varphi$ is not valid in synchronous \teamltl.
	Likewise, the second argument can be used to show that the assumption that $\varphi$ is not valid in synchronous \teamltl implies that $\varphi$ is not valid in $\forall\teamltl$.
\end{proof}

In the following section, we establish that the connection between synchronous $\teamltl$ and $\exteamltl$/$\ateamltl$ is more profound than just a connection between the problems of satisfiability and validity.
We show how model checking of extensions of synchronous $\teamltl$ can be efficiently embedded into $\exteamltl$ and $\ateamltl$.
These results imply that the model checking problem of extensions of $\exteamltl$ and $\ateamltl$ are at least as hard as the corresponding problem for synchronous $\teamltl$. The same holds for the validity problem of $\exteamltl$-extensions and for the satisfiability problem of $\ateamltl$-extensions. However, we conjecture that the latter problems for $\exteamltl$ and $\ateamltl$ are harder than for synchronous $\teamltl$, due to the alternation of quantification (between multisets of traces and \tefs) that is taking place.

\subsection{\texorpdfstring{\boldmath}{}Simulating synchronous TeamLTL with fragments of TeamCTL\texorpdfstring{$^*$}{\unboldmath}}\label{sec:synchTeamLTL-embedding}

We show how to embed synchronous $\teamltl$ into extensions of different subfragments of $\teamctls$: $\exteamltl$, $\ateamltl$, $\exteamctl$ and $\ateamctl$.
This is done by using and expanding on the idea from Subsection \ref{subsec:tefproperties} to use a proposition $o$ with alternating truth values on all traces of a team to track progress. 
We define translations $\varphi\mapsto\varphi^+$ and $\varphi\mapsto\varphi^-$ from synchronous $\teamltl$ to fragments of $\teamctls$ that are used in the embeddings. 
The translations are designed such that $T\models_s \varphi$ if and only if $T_o\models_{\exists} \varphi^+$ (resp. $T_o\models_{\forall} \varphi^-$), where $T_o$ is obtained from $T$ by introducing a fresh alternating proposition $o$. 
Some of the translations additionally preserve satisfiability, i.e., $\varphi$ is satisfiable if and only if $\varphi^+$ (resp. $\varphi^-$) is.

Let us now formalise our results a bit more.
Given a set of traces $T$ over $\ap$, let $T_o$ for $o \notin \ap$ be the set of traces $\{\,t \mid t\upharpoonright_{\ap} \in T \text{ and } o \in t[i] \text{ iff } i \mod 2 = 0\,\}$  over $\ap\uplus\{o\}$.
Here, we use $t\upharpoonright_{\ap}$ to denote the restriction of $t$ to $\ap$.
\begin{lemma}\label{lem:synchTeamLTL-to-fragments}
	Given a synchronous $\teamltl$ formula $\varphi$, one can construct in time linear in $|\varphi|$  a formula $\varphi^+$ in $\exteamltl$ (resp., $\exteamctl($$\bor$$)$) and $\varphi^-$ in $\ateamltl($$\bor$$,\NE)$ \mbox{(resp., $\ateamctl(\subseteq)$)} such that for all multisets of traces $T$:
\[
	T \models_s \varphi \text{ iff } T_o \models_\exists \varphi^+ \quad\text{and}\quad T \models_s \varphi \text{ iff } T_o \models_\forall \varphi^-.
\]
\end{lemma}

\begin{proof}
First, we provide an embedding into $\exteamltl$ and $\ateamltl$.  
For $\exteamltl$, we use the formula $\phisynch'$ from Subsection \ref{subsec:tefproperties}.
It ensures that the existentially quantified \tef is synchronous and therefore progresses on a set of traces in the same way a set of traces would make progress in the synchronous setting.
We translate a synchronous $\teamltl$ formula $\varphi$ into $\exteamltl$ in the following way:
$\varphi^+ := \varphi \land \phisynch'$.

For $\ateamltl$, we use a dual approach.
Rather than identifying a synchronous \tef, we instead eliminate all non-synchro\-nous ones.
We make use of the formula
\[	
 \phioff \dfn \big((\NE \land o\land \X o) \lor \top\big) \bor \big((\NE \land \neg o\land \X \neg o) \lor \top\big).
 \]
The formula $\F \phioff$ expresses that in the current \tef, there is a \textit{defect} where some of the traces in the set of traces do not move for one step.
Using this, we can rule out all non-synchronous \tefs from the universal quantifier.
Our translation is:
$\varphi^- := \varphi \lor \F \phioff$.  

Now we consider the embedding into $\exteamctl$ and $\ateamctl$. For an embedding into $\teamctl$, we have to expand on the ideas used for $\teamltl$ further.
We start with the embedding into $\exteamctl$.
Compared to the embedding into $\exteamltl$ where the time evaluation function is constant throughout the formula and thus can be checked for synchronicity via $\phisynch'$, we have to deal with newly quantified time evaluation functions for each operator in the embedding into $\exteamctl$.
This is done by enforcing synchronicity in the translation of every operator.
In particular, we define the translation $\varphi^+$ inductively such that all modalities in a formula are replaced and atomic propositions and Boolean connectives are kept unchanged.
For the non-trivial cases, $\varphi^+$ is defined as follows:
{\allowdisplaybreaks
\begin{align*}
	(\F \varphi)^+ &\dfn [\dep(o)\Um_\exists (\varphi^+ \land\dep(o))],\\
	(\Gm \varphi)^+ &\dfn [(\varphi^+\land\dep(o))\W_\exists \bot],\\
	(\varphi\Um\psi)^+ &\dfn [(\varphi^+\land\dep(o))\Um_\exists(\psi^+\land \dep(o))],\\
	(\X \varphi)^+ &\dfn \X_\exists \big(\dep(o) \land \varphi^+\big),\\
	(\varphi\W\psi)^+ &\dfn [(\varphi^+\land\dep(o))\W_\exists(\psi^+ \land \dep(o))].
\end{align*} 
}%
This yields a translation from a synchronous $\teamltl$ formula $\varphi$ into $\exteamctl(\bor)$ (note that dependence atoms can be defined using $\bor$).

For the embedding into $\ateamctl(\subseteq)$, we use the same ideas as for the embedding into $\exteamctl(\bor)$.
The only difference here is that we have to make use of the universally quantified instead of the existentially quantified modalities in the inductive definition of $\varphi^-$.
For the non-trivial cases, $\varphi^-$ is given as follows:
\begin{align*}
	(\F \varphi)^- &\dfn [\top\Um_\forall (\varphi^- \lor o \subseteq \neg o)]\\
	(\Gm \varphi)^- &\dfn [\varphi^-\W_\forall o \subseteq \neg o]\\
	(\X \varphi)^- &\dfn \X_\forall \big(o \subseteq \neg o \lor \varphi^-\big)  \\
	(\varphi\Um\psi)^- &\dfn [\varphi^-\Um_\forall(\psi^-\lor o \subseteq \neg o)]\\
	(\varphi\W\psi)^- &\dfn [\varphi^-\W_\forall(\psi^- \lor o \subseteq \neg o)]
\end{align*} 
In each of these cases, it is straightforward to see from the arguments presented in the translation that
\[
T \models_s \varphi \text{ iff } T_o \models_\exists \varphi^+ \quad\text{and}\quad T \models_s \varphi \text{ iff } T_o \models_\forall \varphi^-.\qedhere
\]
\end{proof}
Apart from the previous translation and the corresponding theorem, we also make use of the synchronous $\exteamctl$ modalities in the proof of Theorem~\ref{thm:undecidable}.
There, we use $[\varphi \Um_\sigma \psi]$ for $[(\varphi \land \dep(o)) \Um_\exists (\psi \land \dep(o))]$, $\X_\sigma \varphi$ for $\X_\exists (\varphi \land \dep(o))$ etc.

Lemma~\ref{lem:synchTeamLTL-to-fragments} can be used to reduce the model checking and satisfiability problems for synchronous \teamltl to the corresponding problems for different variants of \teamltl and \teamctl.
\begin{theorem}\label{thm:synchltl-embed}
	Model checking for synchronous $\teamltl$ can be reduced in linear time in the given formula length and model to model checking for $\exteamltl$, $\ateamltl($$\bor$,$\NE)$, $\exteamctl($$\bor$$)$ and $\ateamctl(\subseteq)$.
\end{theorem}
\begin{proof}
	The construction sketched in Subsection \ref{subsec:tefproperties} can be used to transform a Kripke structure generating a multiset of traces $T$ into a Kripke structure generating the multiset of traces $T_o$.
	Then, we immediately obtain the theorem using Lemma~\ref{lem:synchTeamLTL-to-fragments}.
\end{proof}
\begin{theorem}\label{thm:synchltl-sat-embed}
	Satisfiability for synchronous $\teamltl$ can be reduced in linear time in the given formula length to satisfiability $\exteamctl($$\bor$$)$ and $\ateamctl(\subseteq)$.
	Assuming fairness of \tefs, this also holds for $\exteamctl(\dep)$.
\end{theorem}
\begin{proof}
	Compared to the proof of Theorem \ref{thm:synchltl-embed}, this proof is not as direct.
	Since we cannot enforce the parity condition on $o$ using a Kripke structure, we have to make sure that the formula in the translation is unsatisfiable by multisets of traces not satisfying the condition.

	For $\exteamctl$ and $\ateamctl$, we can enforce the parity condition by adding an additional conjunct to the translation.
	Consider the following formulae:
	\begin{align*}
		\psisynch &\dfn (o\land\X_\exists \lnot o)\bor (\lnot o\land\X_\exists o),\\
		\allowbreak
		\psisynch' &\dfn \big(o\land\X_\forall (\lnot o\lor o\subseteq \neg o) \big)\lor \big(\lnot o\land\X_\forall (o \lor o\subseteq \neg o )\big).
	\end{align*}
	\noindent By using these formulae, we can express the alternation on $o$ using $\teamctl$ modalities. The formula $\Gm_\exists \psisynch$ expresses this property directly.
	The formula $o\land \Gm_\forall \psisynch'$ is a variant of this formula that uses $\subseteq$ instead of $\bor$ and universal quantification instead of existential quantification.
	Thus, we have a formula that expresses the alternation of $o$ in $\ateamctl(\subseteq)$ and $\exteamctl(\bor)$.
	Additionally, there is a variant that assumes fairness for \tefs.
	Consider the following $\exists\teamctl$ formula:
	\begin{align*}
		\psisynch'' &\dfn (o\land\X_\exists \lnot o)\lor (\lnot o\land\X_\exists o)
	\end{align*}
	If we impose fairness for time evaluation functions, the formula $o\land \Gm_\exists \psisynch''$ yields the same effect as $\Gm_\exists \psisynch$ and $o\land \Gm_\forall \psisynch'$.

	We translate a synchronous TeamLTL formula $\varphi$ as follows:
	\begin{align*}
		&\exteamctl(\bor) \qquad &\varphi\mapsto \varphi^+\land o \land \Gm_\exists\psisynch \\
		&\exteamctl(\dep) \qquad &\varphi\mapsto \varphi^+\land o \land \Gm_\exists\psisynch'' \\
		&\ateamctl(\subseteq) \qquad &\varphi\mapsto \varphi^-\land o \land \Gm_\forall\psisynch'
	\end{align*}
	where $\varphi^+$ and $\varphi^-$ are defined as in Lemma~\ref{lem:synchTeamLTL-to-fragments}.
	Using the arguments presented above, it is straightforward to see that $\varphi$ is satisfiable in synchronous TeamLTL if and only if its translation is satisfiable in the respective variant of $\teamctl$.
\end{proof}
Note that Theorem \ref{thm:synchltl-sat-embed} does not include a reduction for $\exteamltl$ or $\ateamltl$.
For $\exteamltl$, this is due to the fact that Theorem~\ref{thm:synchltl-sat-validity} already gives us a stronger result than the reduction underlying Lemma~\ref{lem:synchTeamLTL-to-fragments}.
For \ateamltl, it proves to be difficult to make the formula unsatisfiable by sets of traces not satisfying the parity condition.
This is because the formula has to behave correctly for all \tefs, and thus, to establish the alternation on $o$, no assumption about a \tef's progress can be made.

Using Theorem~\ref{thm:synchltl-embed} and the $\Sigma^0_1$-hardness of the model checking problem for synchronous $\teamltl(\varovee,\subseteq)$ \cite[Theorem~2]{DBLP:journals/corr/abs-2010-03311}, we obtain the following undecidability result.

\begin{corollary}
	The model checking problem for $\exteamctl(\varovee,\subseteq)$ is $\Sigma^0_1$-hard.
\end{corollary}

\section{\NoCaseChange{TeamCTL}\texorpdfstring{$(\varovee)$}{(disjunction)} is highly undecidable}

We show how to obtain high undecidability by encoding recurrent computations of \emph{non-deterministic 2-counter machines (N2C)}.	
A non-deterministic 2-counter machine $M$ consists of a list $I$ of $n$ instructions that manipulate two counters ($\ell$eft and $r$ight) $C_\ell$ and $C_r$. All instructions are in one of the following three forms:
\begin{align*}
	&C_a^+\ \textbf{goto}\ \{j, j'\},\\
        &C_a^-\ \textbf{goto}\ \{j, j'\},\\
	&\textbf{if}\ C_a=0\ \textbf{goto}\  j\ \textbf{else\ goto}\  j',
\end{align*}
where $a\in\{\ell,r\}$, $0\leq j,j' < n$. A \emph{configuration} is a tuple $(i,j,k)$, where $0\leq i < n$ is the next instruction to be executed, and $j,k\in \N$ are the current values of the counters $C_\ell$ and $C_r$. 
The execution of the instruction $i\colon C_a^+\ \textbf{goto}\ \{j, j'\}$ ($i\colon C_a^-\ \textbf{goto}\ \{j, j'\}$, resp.) increments (decrements, resp.) the value of the counter $C_a$ by $1$. 
The next instruction is selected nondeterministically from the set $\{j, j'\}$. 
The instruction $i\colon \textbf{if } C_a=0\ \textbf{goto}\ j,\ \textbf{else\ goto}\ j'$ checks whether the value of the counter $C_a$ is currently $0$ and proceeds to the next instruction accordingly. 
The \emph{consecution relation} of configurations is defined as usual. 
A \emph{computation} is an infinite sequence of consecutive configurations starting from the initial configuration $(0,0,0)$. A computation is \emph{$b$-recurring} if the instruction labelled $b$ occurs infinitely often in it.
\begin{theorem}[\cite{AlurH94}]\label{thm:countermachine}
Deciding whether a given nondeterminis\-tic $2$-counter machine has a $b$-recurring computation for a given label $b$ is $\Sigma_1^1$-complete.
\end{theorem}


We reduce the existence of a $b$-recurring computation of a given N2C machine $M$ and an instruction label $b$ to the model checking problem of $\exteamctl($$\bor$$)$. 

\begin{theorem}\label{thm:undecidable}
	Model checking for $\exteamctl($$\varovee$$)$ is $\Sigma_1^1$-hard.
\end{theorem}
\begin{proof}
		Let $I$ be a given set of instructions of a $2$-counter machine $M$ with the set of labels $\mathbb I\coloneqq\{i_1,\dots,i_n\}$ for $n\in\mathbb N$, and an instruction label $b\in\mathbb I$. 
		We construct a $\teamctl($$\varovee$$)$-formula $\varphi_{I,b}$ and a Kripke structure $\KM_I$ such that
\begin{equation}\label{thm:undecidable_eq1}
	\traces(\KM_I) \models_\exists \varphi_{I,b}
	\text{ iff $M$ has a $b$-recurring computation.}
\end{equation}
	\begin{figure*}[t]
		\centering
		\includegraphics[width=.7\linewidth]{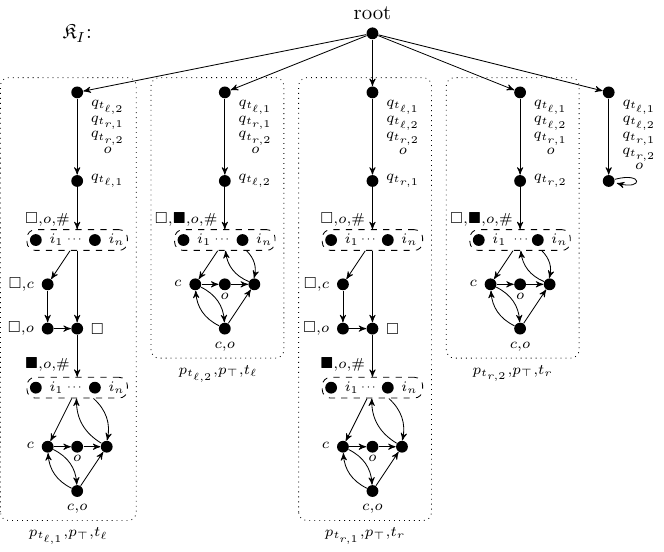}
		\caption{Kripke structure $\KM_I$ used in proof of Theorem~\ref{thm:undecidable}. Dotted boxes mean that the propositions below them are labelled in every state within the box. Dashed boxes with states having labels $i_j$ for $1\leq j\leq n$ simplify the presentation as follows: the incoming/outgoing edges to these boxes are connected with every vertex in the dashed box; the propositions ($\Box,o,\#,\blacksquare$) above dashed boxes are labelled at every state in the box. We call the trace generated via the rightmost part $\pi_d$.}\label{fig:KMI-N2C}
	\end{figure*}
From $\KM_I$ one can obtain all sequences of configurations (even those which are not consecutive computations) for the machine $M$. 
The formula $\varphi_{I,b}$ then allow us to pick some particular traces generated from the structure. 
This essentially corresponds to an existential quantification of the computation. 
The Kripke structure $\KM_I$ is depicted in Figure~\ref{fig:KMI-N2C}.

	Intuitively, the structure is partitioned into five parts.
	The two left-most parts (see Figure ~\ref{fig:KMI-N2C}) encode values of the left counter, the two following parts encode values of the right counter, and the right-most part is a ``dummy trace'' $\pi_d$ (we come to an explanation for $\pi_d$ a bit later). 
    Every trace that is different from $\pi_d$ has a proposition $p_\top$ labelled in every state after the root. 
	We are using four \emph{types} of traces: two trace types $t_{\ell,1},t_{\ell,2}$ (the first and second from left in Figure ~\ref{fig:KMI-N2C} for counter $C_\ell$ and two trace types $t_{r,1},t_{r,2}$ (the third and fourth from left in Figure~\ref{fig:KMI-N2C}) for counter $C_r$. 
	Intuitively, a trace of type $t_{s,2}$ encodes the value for the counter $C_s$ during the whole computation, for $s\in\{r,\ell\}$. See Figure \ref{fig:trace-counter} for an example on how the value of the counter is encoded in a trace by the numbers of points labelled by $c$ between consecutive points labelled by $\#$.
	Our construction ensures that a selected trace of type $t_{s,1}$ is always one step ahead of the selected trace of type $t_{s,2}$. 
	We enforce that they have the same $c$ labelling and therefore the counter value is carried from one $\#$ position to the next (subject to an increment/decrement operation). 
	Such trace pairs are also called $\ell$-traces (or $r$-traces, respectively) and globally have a proposition $t_\ell$ (resp., $t_r$) labelled in their non-root states. 
	Each such trace type $t\in\mathbb T\coloneqq\{t_{\ell,1},t_{\ell,2},t_{r,1},t_{r,2}\}$ also has a proposition $p_t$ that is globally true everywhere (except in the root) and another proposition $q_t$ that is true in the second state while the other three $q_{t'}$ for $t'\in\mathbb T\setminus\{t\}$ are true in the first state (this is used for identification purposes).
	These trace pairs are used to simulate incrementing, resp., decrementing the value of the respective counter.
	A value $m\in\mathbb N$ of a counter is encoded via a $\#$-symbol divided sequence of states in a trace where $m$ states contain a proposition $c$. Consecutive values of a counter is encoded in consecutive $\#$-symbol divided sequences.
	This is depicted in Figure~\ref{fig:trace-counter}.

\begin{figure*}[t]\centering\resizebox{1\linewidth}{!}{%
		\begin{tikzpicture}[
		c/.style={label={90:$c$},circle,fill=black,inner sep=1mm},
		c2/.style={label={270:$c$},circle,fill=black,inner sep=1mm},
		c-o/.style={label={90:$c,o$},circle,fill=black,inner sep=1mm},
		c-o2/.style={label={270:$c,o$},circle,fill=black,inner sep=1mm},
		s1/.style={label={90:$\#,o$},circle,fill=black,inner sep=1mm},
		s2/.style={label={270:$\#,o$},circle,fill=black,inner sep=1mm},
		empty/.style={label={90:},circle,fill=black,inner sep=1mm},
		empty-o/.style={label={90:$o$},circle,fill=black,inner sep=1mm},
		empty-o2/.style={label={270:$o$},circle,fill=black,inner sep=1mm}]

		\begin{scope}[xshift=-2cm]

		\node[anchor=east] at (-5.5,0) {$t_{\ell,2}$:};
		\foreach \x in {-5,9}{
			\node[] (a\x) at (\x,0) {$\dots$};
		}
		\foreach \x in {-4,2,8}{
			\node[s1] (a\x) at (\x,0) {};
		}
		
		\foreach \x in {1,7}{
			\node[empty] (a\x) at (\x,0) {};
		}
		
		\foreach \x in {0}{
			\node[empty-o] (a\x) at (\x,0) {};
		}
		
		\foreach \x in {-3,-1,3,5}{
			\node[c] (a\x) at (\x,0) {};
		}

		\foreach \x in {-2,4,6}{
			\node[c-o] (a\x) at (\x,0) {};
		}

		\foreach \f in {-5,...,8}{
			\pgfmathsetmacro{\t}{int(round(\f+1)}
			
			\path[-stealth'] (a\f) edge (a\t);
		}
			
		\end{scope}
		
		\node[anchor=east] at (-5.5,-1) {$t_{\ell,1}$:};
		\foreach \x in {6,-5,11}{
			\node[] (b\x) at (\x,-1) {$\dots$};
		}
		\foreach \x in {-4,2,10}{
			\node[s2] (b\x) at (\x,-1) {};
		}

		\foreach \x in {1,9}{
			\node[empty] (b\x) at (\x,-1) {};
		}
		
		\foreach \x in {8}{
			\node[empty-o2] (b\x) at (\x,-1) {};
		}
		
		\foreach \x in {-3,-1,3,5,7}{
			\node[c2] (b\x) at (\x,-1) {};
		}

		\foreach \x in {-2,0,4,6}{
			\node[c-o2] (b\x) at (\x,-1) {};
		}

		\foreach \f in {-5,...,10}{
			\pgfmathsetmacro{\t}{int(round(\f+1)}
			
			\path[-stealth'] (b\f) edge (b\t);
		}

		\foreach \f/\t in {-4,...,2}{
			\pgfmathparse{int(round(\t+6))}
			\path[-, dashed, gray] (b\f) edge (a\pgfmathresult);
		}
		
		\end{tikzpicture}}
		\caption{Simulation of a N2C via traces as used in the proof of Theorem~\ref{thm:undecidable}. For presentation reasons, some of the labelled propositions have been omitted. The excerpt of the traces shows that the counter $C_\ell$ is incremented first from 3 to 4 and then to 5. Note that $t_{\ell,2}$ encodes the current value for the counter $C_\ell$. The dashed lines denote which worlds of the respective trace stay in synch.}\label{fig:trace-counter}
	\end{figure*}
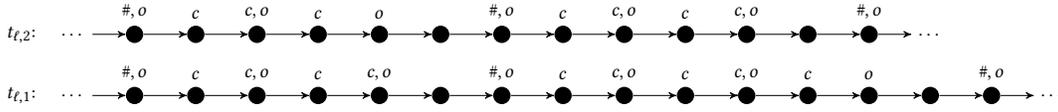%
	The alternating $o$-labels on $\KM_I$-states in combination with strict monotonicity of \tefs, allow for the definition of variants of $\Um$- $\F$-, $\Gm$- and $\X$-operators that are synchronously evaluated in our setting:	
	\begin{align*}
		[\varphi\Um_\sigma\psi]&\coloneqq[\dep(o)\land\varphi\Um_\exists\dep(o)\land\psi],\\
		\F_\sigma\varphi&\coloneqq [\top\Um_\sigma \varphi],\\
		\Gm_\sigma\varphi&\coloneqq\Gm_\exists(\dep(o)\land\varphi), \\
		\X_\sigma\varphi&\coloneqq \X_\exists(\dep(o)\land\varphi).
	\end{align*}

Initially, $\traces(\KM_I)$ contains all possible combinations and sequences of counter modifications. 
	Intuitively, we sort these traces into five groups; traces of types in $\mathbb T$ plus the dummy trace. 
	In the first step, we need to cut this tremendous number of traces down to four traces (plus the dummy trace).
	This is realised by the formula $\varphi_{I,b}$ that is a split into two subformulae (after we synchronously make one step):
\begin{align}\label{eq:Ib}
	\varphi_{I,b} &\coloneqq \X_\sigma(p_\top\lor(\varphi_{\text{struc}}\land\X_\sigma\X_\sigma\varphi_{\text{comp}})).
\end{align}
Note that the dummy trace $\pi_d$ has to be always split to the right side (as it has not $p_\top$ labelled). This prevents an empty split on that side. 
The vast majority of unconsidered traces are assigned to the left side of the split. 

We need to define an auxiliary formula 
\[\varphi_{\text{single}}(S)\coloneqq \Gm_\sigma \bigwedge_{p \in S} \dep(p)\] 
expressing, that a particular trace team agrees on each time step with respect to the set of propositions in $S$. 
Now, we turn to the right-hand side of the split in \eqref{eq:Ib} which is a conjunction of two formulae:
\[
\varphi_{\text{struc}}\!\!\coloneqq\!\!\bigwedge_{t\in\mathbb T}\X_\exists q_t\land \left[(t_\ell\land\varphi_{\text{single}}(S))\!\!\lor\!\!(t_r\land\varphi_{\text{single}}(S))\lor\lnot p_\top\right],
\]
where $S\coloneqq\{\#,c\}\cup \mathbb I$.
The first conjunct of this formula ensures that we are dealing with at least one trace of each type of $t\in\mathbb T$ (the trace $\pi_d$ and every different type $t'\neq t$ is `paused').
The second conjunct splits away $\pi_d$ and groups the traces according to the left and right counter. 
There, it forces that $t_{s,1}$ has the same labelling (with respect to $\#,c$ and instruction labels from $\mathbb I$) as $t_{s,2}$ for $s\in\{r,\ell\}$.
This together with the construction of $\KM_I$ ensures that we are dealing with exactly four traces.

The formula $\varphi_{\text{comp}}$ in \eqref{eq:Ib}, defined by
\[
\varphi_{\text{comp}}\coloneqq\lnot p_\top\lor\left([\Box\Um_\exists \blacksquare\land \theta_\text{valid}]\land\theta_\text{brec}\land i_1\right),
\] 
is used to desynchronise trace $t_{s,1}$ from $t_{s,2}$ (for $s\in\{\ell,r\}$) by exactly one $\#$-interval to enable enforcing the de/increment-operation later.
Again, we first split $\pi_d$ away and then say that the label $b$ is occurring infinitely often via $\theta_\text{brec}\coloneqq\Gm_\exists\F_\exists b$, then desynchronise with the $\Um$ntil operator and say that the computation is valid.
Note that it is crucial to use $\Um_\exists$ for desynchronising the traces (compare Figure~\ref{fig:KMI-N2C}: $t_{s,2}$ has to stay at the first $\#$, while $t_{s,1}$ advances to the second $\#$; so, this is handled via the $\Box$ together with the $\blacksquare$).
The $i_1$ at the end ensures that the first instruction of $I$ is executed first.

	In the following, we define the formulae required for implementing the instructions of the N2C machine; in this context, we write $s=\ell$ and $\bar s=r$ or vice versa, as short-hands for the left and right counter. 
	When we increment/decrement a counter, we need a formula that makes sure that on the traces for the other counter, the counter value stays the same (hence $t_{\bar s,1}$ and $t_{\bar s,2}$ have the same next $\#$-interval with respect to $c$):
	\[
		\text{halt}\coloneqq\X_\sigma[\dep(c)\land\lnot\#\Um_\sigma\#].
	\]
	Now, we can define the formula that states the incrementation of the counter $C_s$.
	Notice that the $c$ in the left part of the synchronous $\Um$ntil makes sure that the counter value differs by exactly one.
	The synchronous $\Um$ntil operator matches the number of $c$-labelled states on both $s$-traces and then verifies that the count for the first trace contains one more $c$:
	\[
		C_s\text-\text{inc}\coloneqq 
		\bigl(\X_\sigma[\,c\,\Um_\sigma
			         (p_{t_{s,2}}\land\lnot c)\lor(p_{t_{s,1}}\land c\land\X_\sigma \lnot c)]
		\bigr)
		\lor (\text{halt}\land t_{\bar s}).
	\]
	Symmetrically, we can define a decrement operation:
	\[
		C_s\text-\text{dec}\coloneqq 
		\bigl(\X_\sigma[\,c\,\Um_\sigma
			         (p_{t_{s,2}}\land c\land \X_\sigma \lnot c)\lor(p_{t_{s,1}}\land \lnot c)]
		\bigr)
		\lor (\text{halt}\land t_{\bar s}).
	\]	
	Now, we define the formulae for the possible instructions.
	\begin{description}
		\item[$i\colon C_s^+$ goto $\{j,j'\}$:] The following formula ensures, that we increase the counter $C_s$ and then reach the next $\#$, where either $j$ or $j'$ is uniformly true.
		\[\theta_i\coloneqq C_s\text-\text{inc}\land\X_\sigma([\lnot\#\Um_\exists \#\land (j\varovee j')]\]

		\item[$i\colon C_s^-$ goto $\{j,j'\}$:] As before, jump to the next $\#$ and have $j$ or $j'$ uniformly after decreasing the counter $C_s$.
		\[\theta_i\coloneqq C_s\text-\text{dec}\land\X_\sigma([\lnot\#\Um_\exists \#\land (j\varovee j')]\]
		
		\item[$i\colon$if $C_s=0$ goto $j$, else goto $j'$:] Here, we use the Bool\-ean disjunction for distinguishing the part of the if-then-else construct. 
		Recall that $t_{s,2}$ encodes the current counter value of $C_s$. 
		The halt-formula ensures that the counter values stay the same:
		\begin{align*}\theta_i\coloneqq\X_\sigma
		\biggl(
		\bigl(
			((&q_{t_{s,2}}\land \lnot c)\lor(\lnot q_{t_{s,2}}))
			\land
			[\lnot\#\Um_\exists \#\land j]
		\bigr)
		\varovee\\
		\bigl(
			((&q_{t_{s,2}}\land c)\lor(\lnot q_{t_{s,2}}))
			\land
			[\lnot\#\Um_\exists \#\land j']
		\bigr)
		\biggr)\land\text{halt}.
		\end{align*}
	\end{description}
	Next, we need a formula stating that only the label $j\in\mathbb I$ is true:
	\[
	\text{only}(j)\coloneqq j\land\bigwedge_{i\neq j}\lnot i.
	\]
	Now, let $\theta_{\text{valid}}$ state that each $\#$-labelled position encodes the correct instruction of the N2C machine. 
	Note that due to $\dep(\#)$ the encoding of the traces remains synchronised:
	\[
	\theta_\text{valid}\coloneqq \Gm_\exists\left(\dep(\#)\land\left(\left(\#\land \bigvee_{i<n}(\text{only}(i)\land\theta_i)\right)\lor\neg\#\right)\right).
	\] 	
 	The length of the formula $\varphi_{I,b}$ is linear in $|\mathbb I|$.

	The direction ``$\Leftarrow$'' of claim~(\ref{thm:undecidable_eq1}) follows by construction of $\varphi_{I,b}$ and $\KM_I$. 
	For the other direction, note that if $\varphi_{I,b}$ is satisfied then this means the label $b$ needs to appear infinitely often and the encoded computation of the N2C machine is valid.
\end{proof}

\section{Translating from \NoCaseChange{TeamCTL}\texorpdfstring{$^*$}{*} to AABA}\label{sec:translation}
In the previous section, we showed that model checking is highly undecidable even for a very restricted fragment of $\teamctls$.
This motivates the search for decidable restrictions.

\subsection{Asynchronous automata and tefs}
We first introduce Alternating Asynchronous B\"uchi Automata.
\begin{definition}[\cite{GutsfeldMO21}]
	Let $M = \{1, 2, \dots, n\}$ be a set of directions and $\Sigma$ an input alphabet. 
	An \emph{Alternating Asynchronous B\"uchi Automaton (AABA)}  is a tuple $\mathcal{A} = (Q,\rho^0,\rho,F)$ where $Q$ is a finite set of states, $\rho^0 \in \mathcal{B}^+(Q)$ is a positive Boolean combination of initial states, $\rho\colon Q \times \Sigma \times M \to \mathcal{B}^+(Q)$ maps triples of control locations, input symbols and directions to positive Boolean combinations of control locations, and $F \subseteq Q$ is a set of final states.
\end{definition}

For the definition of runs of AABA, we need a notion of trees.
Formally, a \emph{tree} $\mathfrak{T}$ is a subset of $\mathbb{N}^{*}$ such that for every node $\mathfrak{t} \in \mathbb{N}^{*}$ and every positive integer $n \in \mathbb{N}$: $\mathfrak{t} \cdot n \in \mathfrak{T}$ implies (i) $\mathfrak{t} \in \mathfrak{T}$ (we then call $\mathfrak{t} \cdot n$ a \emph{child} of $\mathfrak{t}$), and (ii) for every $0 < m < n$, $\mathfrak{t} \cdot m \in \mathfrak{T}$. 
We assume every node has at least one child.
A \emph{path} in a tree $\mathfrak{T}$ is a sequence of nodes $\mathfrak{t}_0 \mathfrak{t}_1 \dots$ such that $\mathfrak{t}_0 = \varepsilon$ and $\mathfrak{t}_{i+1}$ is a child of $\mathfrak{t}_i$ for all $i \in \mathbb{N}$.

\begin{definition}
	A \emph{run} of an AABA $\mathcal{A}$ over $n$ infinite words $w_1,\dots,w_n \in \Sigma^{\omega}$ is defined as a $Q$-labeled tree $(\mathfrak{T},r)$ where $r\colon \mathfrak{T} \to Q$ is a labelling function.
	Additionally, for each $\mathfrak{t} \in \mathfrak{T}$, we have $n$ offset counters $c_1^\mathfrak{t},\dots,c_n^\mathfrak{t}$ starting at $c_i^{\mathfrak{t}} = 0$ for all $i$ and $\mathfrak{t}$ with $|\mathfrak{t}| \leq 1$.
	Together, the labelling function and offset counters satisfy the following conditions: 
	\begin{enumerate}[(i)]
		\item we have $\{\,r(\mathfrak{t}) \mid \mathfrak{t} \in \mathfrak{T}, |\mathfrak{t}| = 1\,\} \models \rho^0$, and 
		\item when node $\mathfrak{t} \in \mathfrak{T} \setminus \{\varepsilon\}$ has children $\mathfrak{t}_1,\dots,\mathfrak{t}_k$, then there is a $d \in M$ such that 
		\begin{enumerate}[(a)]
			\item $c_d^{\mathfrak{t}_i} = c_d^\mathfrak{t} + 1$ and $c_{d'}^{\mathfrak{t}_i} = c_{d'}^{\mathfrak{t}}$ for all $i$ and $d' \neq d$, 
			\item we have $1 \leq k \leq |Q|$, and
			\item the valuation assigning true to $r(\mathfrak{t}_1),\dots,r(\mathfrak{t}_k)$ and false to all other states satisfies $\rho(r(\mathfrak{t}),w_d(c_d^\mathfrak{t}),d)$. 
		\end{enumerate}
	\end{enumerate}
	A run $(\mathfrak{T},r)$ is an \emph{accepting run} iff for every path $\mathfrak{t}_0 \mathfrak{t}_1 \dots$ in $(\mathfrak{T},r)$, a control location $q \in F$ occurs infinitely often.
\end{definition}

We say that $(w_1,\dots,w_n)$ is \emph{accepted} by $\mathcal{A}$ iff there is an accepting run of $\mathcal{A}$ on $w_1,\dots,w_n$.
The \emph{set of tuples of infinite words} accepted by $\mathcal{A}$ is denoted by $\mathcal{L}(\mathcal{A})$.

We also use AABA with a generalised B\"uchi acceptance condition which we call \emph{Generalised Alternating Asynchronous B\"uchi Automata (GAABA)}.
Formally, a GAABA is an AABA in which the set $F = \bigcup F_i$ consists of several acceptance sets $F_i$ and the B\"uchi acceptance condition is refined such that a run $(\mathfrak{T},r)$ is accepting iff every path visits a state in every set $F_i$ infinitely often. 
A GAABA can be translated to an AABA with quadratic blowup similar to the standard translation from Generalised B\"uchi Automata to B\"uchi Automata~\cite{DBLP:books/cu/Demri2016}. For this purpose, we first annotate states with the index $i$ of the next set $F_i$ that requires a visit of an accepting state. We then consecutively move to the states annotated with the next index if such an accepting state is seen. Only those states are declared accepting that indicates that an accepting state in every $F_i$ has been visited.

AABA and tefs model asynchronicity in a similar manner.
In fact, the progress of time in each path of a run in an AABA can be considered as a non-parallel tef over a finite number of traces.
Formally, let $(\mathfrak{T},r)$ be a run of an AABA over $n\in\N$ words. 
For every path $\pi \in (\mathfrak{T},r)$, we denote by $\tau_\pi$ the time evaluation function obtained by setting $\tau_\pi(i)\dfn(c_1^{\pi(i)},\dots, c_n^{\pi(i)})$.
Then, the set $I_{(\mathfrak{T},r)} \dfn \{\,\tau_\pi \mid \pi \text{ is a path in }(\mathfrak{T},r)\,\}$ is the set of \emph{induced \tefs} for a run $(\mathfrak{T},r)$.

We consider the language of AABA where only certain sets of induced \tefs are allowed.
Formally, for an AABA $\mathcal{A}$ and a set of \tefs $\mathit{TE}$, the \emph{language of $\mathcal{A}$ restricted to $\mathit{TE}$} is given by $\mathcal{L}_{\mathit{TE}}(\mathcal{A}) \dfn \{\,w = (w_1, \dots, w_n) \in (\Sigma^{\omega})^n \mid \exists (\mathfrak{T},r) \text{ s.t. } (\mathfrak{T},r) \text{ is an accepting run}$ $\text{ of }\mathcal{A} \text{ on } w \text{ and }  I_{(\mathfrak{T},r)} \subseteq \mathit{TE}\,\}$.
Then, the \emph{$\mathit{TE}$ emptiness problem for an AABA} $\mathcal{A}$ is to check whether $\mathcal{L}_{\mathit{TE}}(\mathcal{A})$ is empty.

In particular, we consider two restrictions on \tefs that were already considered in the literature, $k$-synchronous \tefs and $k$-context bounded \tefs \cite{GutsfeldMO21}.
Given a team $(T,\tau)$ and $k \in \mathbb{N}$, the \tef $\tau$ is \emph{$k$-synchronous} iff for all $i \in \mathbb{N}$ and $t, t' \in T$, $|\tau(i, t) - \tau(i, t')| \leq k$. 
The term $k$-synchronicity indicates that traces are not allowed to diverge by more than $k$ steps during the execution.
For the definition of the second restriction, let $\tau$ be a non-parallel \tef and let 
\[
\switch_{\tau}(i) =
\begin{cases}
1 & ,\text{ if $i>0$ and $\exists t\in T$ s.t. $\tau(i-1, t) = \tau(i, t)$ and} \\
&\text{\hphantom{, if $i>0$ and $\exists t\in T$ s.t. }} \tau(i, t) \neq \tau(i+1, t),\\
0 & ,\text{ otherwise}.
\end{cases}
\]
be a Boolean function that indicates whether a context switch is performed between steps $i$ and $i+1$ in $\tau$. 
A \tef $\tau$ is \textit{$k$-context-bounded} iff $\tau$ is non-parallel and $\switch_{\tau}(i) \leq 1$ for all $i \in \mathbb{N}$ and $\sum_{i \in \mathbb{N}}\switch_{\tau}(i) \leq k$.
The term $k$-context-boundedness states that only a single trace is allowed to progress on each global time step and we can only switch between different traces at most $k$ times.

For AABA, the following holds.
\begin{theorem}[{\cite[Cors.~3.6 \& 3.13, Thms.~3.12 \& 3.17]{GutsfeldMO21}}]\label{theorem:AABA}$\,$
	\begin{enumerate}
		\item The emptiness problem for AABA is undecidable.
		\item The $k$-synchronous emptiness problem for AABA with $n$ traces is $\mathrm{EXPSPACE}$-complete and is $\mathrm{PSPACE}$-complete for fixed $n$.
		\item The $k$-context-bounded emptiness problem for AABA is $(k-2)$-$\mathrm{EXPSPACE}$-complete.
	\end{enumerate}
\end{theorem}
Note that the hardness results in \cite{GutsfeldMO21} were formulated for AAPA (Alternating Asynchronous Parity Automata), not AABA.
However, the proofs rely only on reachability and not on a parity acceptance condition.
Thus, they carry over to AABA as well.

Similar to the restricted semantics for AABA, one can define restricted semantics for \teamctls.
For a team $(T,\tau)$, a $\teamctls$ formula $\varphi$ and a set of \tefs $\mathit{TE}$ with $\tau\in \mathit{TE}$, we write $(T,\tau,i) \models_{\mathit{TE}} \varphi$ (for all timesteps $i\in\N$) to denote that $(T,\tau,i)$ satisfies $\varphi$ under a semantics in which quantifiers $\exists$ and $\forall$ range only over \tefs in $\mathit{TE}$.
This is straightforwardly extended to modes of satisfaction $\models_{*,\mathit{TE}}$ for $* \in \{\exists,\forall\}$.
The \emph{fixed size $\mathit{TE}$ satisfiability problem $\fsSAT$} is then to decide for a given natural number $n$ and a $(\teamctls,\models_{*})$ formula $\varphi$ whether there exists a multiset of traces $T$ with $|T| = n$ such that $T \models_{*,\mathit{TE}} \varphi$. 
The \emph{fixed size $\mathit{TE}$ model checking problem $\fsMC$} is to decide for a Kripke model, a $(\teamctls,\models_{*})$ formula $\varphi$ and natural number $n$ whether all multisets of traces $T \subseteq \traces(\KM)$ of size $n$ satisfy $\varphi$ under $\mathit{TE}$ semantics. 
The \emph{$TE$ path checking problem $\PC$} is defined analogously.
We set $\mathcal{L}_{\mathit{TE}}^n(\varphi) \dfn \{\;T \mid |T| = n \land T \models_{*,\mathit{TE}} \varphi \,\}$ to be the set of finite multisets of traces of size $n$ satisfying $\varphi$ under the $\mathit{TE}$ semantics.

\subsection{Idea behind the translation}
We now present a translation of $\teamctls(\mathcal{S})$ to AABA over a fixed number $n\in\N$ of traces where $\mathcal{S}  \dfn \{\bor,\bneg, \NE, \flatop, \dep{}, \subseteq \}$ is the set of all additional atoms and connectives used in this paper. 
Note that Theorem~\ref{thm:GA} would allow us to drop `$\dep{}$' and `$\subseteq$' from the translation, but the cost would be an exponential blowup, which does not occur in our direct translation.
Our translation is obtained by first constructing a suitable Generalised Alternating Asynchronous B\"uchi Automaton (GAABA) and then translating that GAABA into an AABA as described above.

We start with the intuition of translating $\teamctls(\mathcal{S})$ formulae to GAABA.
We consider explicitly only existential quantification over \tefs and the $\models_{\exists}$ mode of satisfaction. 
Universal quantifiers as well as the other types of satisfaction are reduced to their existential counterparts by using Boolean negation.
This translation is based on the classical Fischer-Ladner construction, which translates LTL formulae into non-deterministic B\"uchi automata \cite{DBLP:books/cu/Demri2016}. 
We construct a GAABA $\mathcal{A}_{\psi}$ inductively over the quantification depth such that $\mathcal{A}_{\psi}$ has an accepting run over the input multiset of traces under $\mathit{TE}$ semantics if and only if there exists a \tef such that $\psi$ is satisfied by the input multiset of traces under $\mathit{TE}$ semantics.
This allows us to apply decidability results for AABA under restricted semantics, in particular those of \cite{GutsfeldMO21}, to decision problems for $\teamctls(\mathcal{S})$.

Let us explain this construction in turn. 
In general, we perform the standard Fischer-Ladner construction for LTL, i.e., we annotate states with the formulae that should hold whenever an accepting run visits the respective states. 
The transition function is then defined such that the requirements for the successor induced by these formulae are met. 
Finally, we choose accepting sets such that in an accepting run, $[\psi_1\Um\psi_2]$ formulae are either not required at a given time step or the run must eventually visit a state containing $\psi_2$. The indices associated with the formulae indicate the team traces we are considering for that formula.
There are three additional problems not covered by the standard construction: asynchronicity, quantifiers, and non-standard logical operators not present in LTL, i.e., the split operator and additional team semantics atoms.

In a GAABA, time evaluation functions with asynchronous progress on different traces can be modelled straightforwardly by asynchronous moves. 
In every step, we non-deterministically pick the traces to be advanced, ensure that the atomic propositions in the next state of $\mathcal{A}_{\varphi}$ match those given by the next index of that trace and maintain the atomic propositions for all other traces.
Without loss of generality, we assume that only one trace is advanced in each time step by a given \tef in the formal description of the construction and it will be clear how to extend the construction to general \tefs.

In order to handle existential quantifiers, we use conjunctive transitions.
For every existential subformula $\exists \psi$ contained in a state $q$, we guess an initial state of $\mathcal{A}_{\psi}$ which agrees with $q$ on the atomic propositions and then conjunctively proceed from that state according to the transition function of $\mathcal{A}_{\psi}$.
For universal quantifiers and the $\models_{\forall}$ mode of satisfaction, we proceed as follows.
Formulae $\forall \psi$ are rewritten using the well-known equivalence $\forall \psi \equiv \bneg \exists \bneg \psi$.
Then, Boolean negations of existentially quantified formulae $\bneg \exists \psi$ are handled similarly to existentially quantified formulae.
The main difference is that we transition to the complement of $\mathcal{A}_{\psi}$ instead of $\mathcal{A}_{\psi}$ itself when encountering a formula $\bneg \exists \psi$ in a state.

Finally, let us discuss the split operator and team semantics atoms.
For the former, we use indexed subformulae to keep track of which traces we are analysing with respect to each subformula.
In the initial states, we work with the full index set $\{1,\dots,n\}$ for the formula $\varphi$, and whenever we encounter a split $(\psi_1 \lor \psi_2)$, we non-deterministically partition the given index set into two sets and check each $\psi_i$ only against its respective index set.
Using this information, we can directly infer whether a dependence atom $\dep(\psi_1, \dots, \psi_n, \varphi)$ holds in that state by checking whether all traces agreeing on the formulae $\psi_1, \dots, \psi_n$ also agree on $\varphi$. Likewise, the atoms $\NE$, $\flatop$ and $\subseteq$ are handled by explicitly checking the semantics for each state.
Notice again that we only allow propositional formulae in dependence atoms, so the semantics of these atoms depend only on the current state.

\subsection{Formal translation}

We now describe the construction more formally.
As mentioned, we only treat existential quantifiers explicitly and use negation to handle universal quantifiers.
In order to achieve this goal, we rewrite each subformula $\forall \psi$ as $\bneg \exists \bneg \psi$.
Thus, we can transform every $\teamctls(\mathcal{S})$ formula into an equivalent formula without universal quantifiers and assume that formulae are in $\exteamctls(\mathcal{S})$ in the rest of this section.

Our translation is based on the notions of \textit{indexed Fischer-Ladner-closure}, \textit{consistent subsets of the indexed Fischer-Ladner-closure} and the \textit{local transition relation}.
We introduce these notions first.
For a natural number $n$ (this will correspond to the number of traces later), we set $[n] \coloneqq \{1, \dots, n\}$ and $I\coloneqq 2^{[n]}$. 
For an index set $M \in I$, let $\SP(M) = \{\,(M_1, M_2) \mid M = M_1 \uplus M_2\,\}$ be the set of possible splits. Notice that $M_1$ or $M_2$ can be the empty set here.
\begin{definition}
	Let $n\in\mathbb N$ and  $I\coloneqq 2^{[n]}$. 
	For an $\exteamctls(\mathcal{S})$ formula $\varphi$, let $\Sub(\varphi)$ be the set of subformulae of $\varphi$, closed under negation, in which double negations are cancelled and $\bneg\bneg \psi$ is therefore identified with $\psi$.
	The \emph{indexed Fischer-Ladner-closure} of $\varphi$ over $n$ traces is given by $\cl(\varphi) = \Sub(\varphi) \times I$.
\end{definition}
\begin{definition}
	Let $n\in\mathbb N$ and $\varphi$ be an $\exteamctls(\mathcal{S})$ formula, and $\cl(\varphi)$ be the indexed Fischer-Ladner-closure over $n$ traces. 
	A set $S \subseteq \cl(\varphi)$ is \emph{consistent} iff the following holds:
	\begin{enumerate}
		\item Let $\psi$ be a propositional formula such that $\psi,\lnot\psi$ occur in $\cl(\varphi)$. Then we have that 
		\begin{itemize}
			\item for all $i\in[n]$: $(\psi, \{i\}) \in S$ iff $(\neg \psi, \{i\}) \notin S$.
			\item for all $M\in I$: $(\psi, M) \in S$ iff $\forall j \in M: (\psi, \{j\}) \in S$. 
			\item for all $M\in I$: $(\lnot\psi, M) \in S$ iff $\forall j \in M: (\lnot\psi, \{j\}) \in S$. 
		\end{itemize}
		\item If $((\psi_1 \land \psi)_2, M) \in \cl(\varphi)$, then $((\psi_1 \land \psi_2), M) \in S$ iff $(\psi_1, M) \in S$ and $(\psi_2,M) \in S$.
		\item If $((\psi_1 \lor \psi)_2, M) \in \cl(\varphi)$, then $((\psi_1 \lor \psi_2), M) \in S$ iff $(\psi_1, M_1) \in S$ and $(\psi_2, M_2) \in S$ for some $(M_1, M_2) \in \SP(M)$.
		\item If $(\psi_1 \bor \psi_2, M) \in \cl(\varphi))$, then $(\psi_1 $$\bor$$ \psi_2, M) \in S$ iff $(\psi_1, M) \in S$ or $(\psi_2, M) \in S$.
		\item Let $\psi \in \cl(\varphi)$. Then $\psi \in S$ iff $\bneg \psi \not\in S$.
		\item For all $(\dep(\varphi_1, \dots, \varphi_k, \psi), M) \in \cl(\varphi)$ it holds that
		 $(\dep(\varphi_1, \dots,$ $\varphi_k, \psi), M) \in S$ iff $\forall i, j \in M:$
  \begin{multline*}
(\forall 1 \leq \ell \leq k: (\varphi_\ell, \{i\}) \in S \iff  (\varphi_\ell, \{j\}) \in S)\\
  \text{implies } ((\psi, \{i\}) \in S \iff (\psi, \{j\}) \in S).
    \end{multline*}
  
  \label{rule:dep}
		\item If $(\NE, M) \in \cl(\varphi)$, then $(\NE, M) \in S$ iff $M \neq \emptyset$.
		\item If $(\flatop \psi, M) \in \cl(\varphi)$, then $(\flatop \psi, M) \in S$ iff $\forall i \in M: (\psi, \{i\}) \in S$.
		\item For all $(\varphi_1 \cdots \varphi_k \subseteq \psi_1 \cdots \psi_k, M) \in \cl(\varphi)$ it holds that $(\varphi_1 \cdots \varphi_k \subseteq \psi_1 \cdots \psi_k, M) \in S$ iff for all $i \in M$ there exists $j \in M$ such that
  \begin{itemize}
  \item if $(\varphi_\ell, \{i\}) \in S$ then $(\psi_\ell, \{j\}) \in S$, and
  \item if $(\lnot\varphi_\ell, \{i\}) \in S$ then $(\lnot\psi_\ell, \{j\}) \in S$,
  \end{itemize}
   for $1 \leq \ell \leq k$
  \label{rule:subseteq}
	\end{enumerate}
	We denote the set of all consistent sets by $\Con(\varphi)$.
\end{definition}
\begin{definition}
	Let $\Sigma\coloneqq2^\ap$. 
	The \emph{local transition relation}
	$\xrightarrow{}$ is a maximal subset of $\Con(\varphi) \times [n]  \times \Sigma \times \Con(\varphi)$ such that for $S \xrightarrow{i,\sigma} S'$ (meaning $(S,i,\sigma,S')\in\,\to$) with $S, S' \in \Con(\varphi)$, $\sigma\in \Sigma$ and $i \in [n]$, the following holds: 
	\begin{itemize}
		\item For $j \neq i$ with $(\psi, \{j\}) \in S$: $(\psi, \{j\}) \in S'$, where $\psi\in\{\,p,\lnot p\mid p\in\ap\,\}$.
		\item For all $p \in \ap: (p, \{i\}) \in S' \iff p \in \sigma$.
		\item If $(\X\psi, M) \in \cl(\varphi)$, then $(\X\psi, M) \in S$ iff $(\psi, M) \in S'$.
		\item If $([\psi_1 \op \psi_2] , M) \in \cl(\varphi)$, then $([\psi_1 \op \psi_2] , M) \in S$ iff $(\psi_2,M) \in S$ or ($(\psi_1,M) \in S$ and $([\psi_1 \op \psi_2] , M) \in S'$).
	\end{itemize}
\end{definition}

This relation roughly denotes whether a state would be a suitable successor of another state according to the transition relation in the standard Fischer-Ladner construction when an input symbol is read on the $i$-th component.
Notice that this relation and the transition function we are about to define progress on a single trace in one step, while our \tefs can progress on multiple steps at once. This is without loss of generality as simultaneous progress on multiple traces can easily be simulated using consecutive transitions over intermediate states.

Using the definitions just described, we can construct a suitable AABA for a $\teamctls(\mathcal{S})$ formula $\varphi$. 
By treating sets of tuples as multisets in the obvious way, we then obtain the following Theorem.
\begin{theorem}\label{thm:aaba-embed}Let $\mathcal{S}=\{\bor,\bneg, \NE, \flatop, \dep{}, \subseteq \}$.
	\begin{enumerate}
		\item For every $(\teamctls(\mathcal{S}),\models_{*})$ formula $\varphi$ and natural number $n$, there is a GAABA $\mathcal{A}_{\varphi}$ such that $\mathcal{L}_{\mathit{TE}}(\mathcal{A}_{\varphi})$ is equivalent to $\mathcal{L}_{\mathit{TE}}^n(\varphi)$ for all sets of \tefs $\mathit{TE}$.
		\item For every $(\teamctls(\mathcal{S}),\models_{*})$ formula $\varphi$ and natural number $n$, there is an AABA $\mathcal{A}_{\varphi}$ such that $\mathcal{L}_{\mathit{TE}}(\mathcal{A}_{\varphi})$ is equivalent to $\mathcal{L}_{\mathit{TE}}^n(\varphi)$ for all 
		sets of \tefs $\mathit{TE}$.
		\item For all sets of \tefs $\mathit{TE}$, if it is decidable whether $\mathcal{L}_{\mathit{TE}}(\mathcal{A})$ is empty for every AABA $\mathcal{A}$, then the finite $\mathit{TE}$ satisfiability and model checking problems and the $\mathit{TE}$ path checking problem for $\teamctls(\mathcal{S})$ are decidable.
	\end{enumerate}
\end{theorem}
\begin{proof}
Since $\teamctls(\mathcal{S})$ formulae can be transformed into $\exteamctls(\mathcal{S})$ formulae as described at the start of this subsection, we assume $\varphi$ to be given as an $\exteamctls(\mathcal{S})$ formula.
We prove the claim by induction on the quantifier depth of $\varphi$ for the following induction hypothesis:
There is a GAABA $\mathcal{A}_{\psi} = (Q_{\psi}, \rho_{\psi}^0, F_{\psi},\Sigma, \rho_{\psi})$ for all maximal quantified subformulae $\psi = \exists \psi'$ and negated quantified subformulae $\psi = \bneg\exists \psi'$ of $\varphi$ such that $(w_1, \dots, w_n) \in \mathcal{L}_{\mathit{TE}}(\mathcal{A}_{\psi})$ iff there is a \tef $\tau \in  \mathit{TE}$ (resp. there is no \tef $\tau \in  \mathit{TE}$) with $(\{w_1,\dots,w_n\}, \tau, 0) \models_{\mathit{TE}} \psi'$.
Here, the automaton for negated quantifiers is obtained by complementing the automaton for non-negated quantifiers.

We denote by $\mathrm{Ext}(q)$ the set of pairs $(\psi, M)$ where $\psi$ is an existentially quantified or negated existentially quantified formula and $M$ is an index set.
Further, we use
$\mathrm{Succs}(q, \sigma, i) = \{\,q' \in Q \mid q \xrightarrow{\sigma,i} q' \,\}$
to denote the possible successor states according to the standard Fischer-Ladner construction.
We then set
\[
Q_{\varphi} = Con(\varphi) \cup \bigcup_{\psi \in \mathrm{Ext}(\varphi)} Q_{\psi},
\]
\[
\rho_{\psi}^0 = \bigvee_{q \in Q_{\varphi} with (\varphi, [n]) \in q} q
\]
and $\Sigma = 2^{AP}$.
For states $q \in Q_{\psi}$ for some subformula $\psi$, we set $\rho_{\varphi}(q, \sigma, i) = \rho_{\psi}(q, \sigma, i)$.
For all other states $q$, we distinguish two cases.
If $\mathrm{Ext}(q) = \emptyset$,
then $\rho_{\varphi}(q, \sigma, i) = \bigvee_{q' \in \mathrm{Succs}(q, \sigma, i)}q'$.
Otherwise, we set $\rho_{\varphi}(q, \sigma, i)\dfn$
\[ 
\bigvee_{q' \in \mathrm{Succs}(q, \sigma, i)}\bigwedge_{(\psi,M) \in \mathrm{Ext}(q)}\bigvee_{Q'' \in \mathrm{Extsuccs}(q, \sigma, i, \psi, M)} \bigwedge_{q'' \in Q''} (q' \land q''),
\] 
where 
\begin{multline*}
\mathrm{Extsuccs}(q, \sigma, i, \psi, M) \dfn \{\,Q' \mid \exists q' \in Q_{\psi}^0:\\
(\psi, M) \in q' \land \forall i \in [n]\forall ap \in AP:(ap,i) \in q \iff\\ (ap,i) \in q' \land Q' \models \rho_{\psi}(q', \sigma, i)  \,\}.	
\end{multline*}

Finally, we set $\mathrm{Finset}([\psi_1 \Um \psi_2], M) \dfn \{\,q \in Q \mid ([\psi_1 \Um \psi_2], M) \notin q \lor (\psi_2, M) \in q\,\}$, $F_{\varphi} \dfn \{\,\mathrm{Finset}([\psi_1 \Um \psi_2], M) \mid ([\psi_1 \Um \psi_2], M) \in cl(\varphi)\,\}$ if $cl(\varphi)$ contains at least one indexed $\U$-formula and $F_{\varphi} = Q_{\varphi}$ otherwise.

This translation allows us to define semantic restrictions generically over \tefs and apply them uniformly to both GAABA and $\teamctls(\mathcal S)$.
We now describe how this can be used to transfer the decidability results of \cite{GutsfeldMO21} for restricted semantics of GAABA to $\teamctls(\mathcal S)$ in order to obtain decidability for the path checking problem and our finite variants of the satisfiability and model checking problems.
If $\mathcal{L}_{\mathit{TE}}(\mathcal{A}_{\varphi})$ can be tested for emptiness, then the finite $\mathit{TE}$ satisfiability problem can be decided directly by this emptiness test.
Analogously, the path-checking problem can be solved by intersecting $\mathcal{L}_{\mathit{TE}}(\mathcal{A}_{\varphi})$ and the input team which can of course be represented by a GAABA too. 
Finally, the finite $\mathit{TE}$ model checking problem for teams of size $n$ can be solved by building $\mathcal{A}_{BN(\varphi)}$ and testing the product with the GAABA accepting all teams of size $n$ over the input Kripke model.
\end{proof}

We note that the undecidability proof for $\teamctl($$\varovee$$)$ model checking of Theorem~\ref{thm:undecidable} relies on the use of infinite teams. 
This directs us to consider teams of fixed size in order to recover decidability.
Finally, applying Theorem~\ref{theorem:AABA} to Theorem~\ref{thm:aaba-embed}, and utilising Theorem \ref{thm:GA}, we obtain the following. Here $\mathrm{ALL}$ denotes the set of all generalised atoms.
\begin{corollary}\label{cor:teamctls-k-synch-k-bounded} Let $\mathit{TE}$ be the set of $k$-synchronous or $k$-context-bounded \tefs and $\mathcal{S}=\{\bor,\bneg, \NE, \flatop, \dep{}, \subseteq \}$.
	\begin{enumerate}
		\item The problems $\fsSAT,\fsMC,$ and $\PC$ for the logics $\teamctls(\mathcal{S},\mathrm{ALL})$ are decidable.
		\item For a fixed formula $\varphi$ and fixed $n, k \in \mathbb{N}$, the fixed size $k$-synchronous or $k$-context-bounded model checking over $n$ traces for $\teamctls(\mathcal{S})$ is decidable in polynomial time. 
	\end{enumerate}
\end{corollary}
The last item follows from the fact that for fixed parameters, the automata $\mathcal{A}_{\varphi}$ are of constant size. The translations described in \cite[Theorem 3.11, Corollary 3.16]{GutsfeldMO21} can be used to obtain equivalent B\"uchi automata of constant size. The emptiness test on the product with the automaton accepting the input traces is then possible in polynomial time.

\begin{table}[t]
	\centering
	\begin{tabular}{p{4.1cm}c}\toprule
	Model Checking Problem for  & Complexity \& Reference\\\midrule
		$\exteamltl($$\varovee$$,\subseteq)$ & $\Sigma^0_1$-hard \hfill{\footnotesize (Cor.~\ref{thm:synchltl-embed} \& \cite[Thm.~2]{DBLP:journals/corr/abs-2010-03311})}\\
		$\ateamltl($$\varovee$$,\subseteq,\NE)$ & $\Sigma^0_1$-hard \hfill{\footnotesize (Cor.~\ref{thm:synchltl-embed} \& \cite[Thm.~2]{DBLP:journals/corr/abs-2010-03311})}\\
		$\ateamctl($$\varovee$$,\subseteq)$ & $\Sigma^0_1$-hard \hfill{\footnotesize (Cor.~\ref{thm:synchltl-embed} \& \cite[Thm.~2]{DBLP:journals/corr/abs-2010-03311})}\\
		$\exteamctl($$\varovee$$)$ & $\Sigma^1_1$-hard \hfill{\footnotesize (Thm.~\ref{thm:undecidable})}\\
		$\teamctls(\mathcal{S},\mathrm{ALL})$ for $k$-syn\-chronous or $k$-context-bounded tefs & decidable \hfill{\footnotesize (Thm.~\ref{thm:GA}, Cor.~\ref{cor:teamctls-k-synch-k-bounded})}\\
		$\teamctls(\mathcal{S})$ for $k$-synchro\-nous or $k$-context-bounded tefs, where $k$ and the number of traces is fixed & polynomial time \hfill{\footnotesize (Cor.~\ref{cor:teamctls-k-synch-k-bounded})}\\\bottomrule
	\end{tabular}
	\caption{Complexity results overview. The $\Sigma^0_1$-hardness
	results follow via embeddings of synchronous $\teamltl$, whereas
	the $\Sigma^1_1$-hardness truly relies on asynchronity.
	$\mathrm{ALL}$ is the set of all generalised atoms and $\mathcal{S}=\{\bor,\bneg, \NE, \flatop, \dep{}, \subseteq \}$.}\label{tbl:overview}
\end{table}

\section{Conclusion} 
In this paper, we revisited temporal team semantics and introduced quantification over \tefs in order to obtain fine-grained control of asynchronous progress over different traces.
We discussed required properties for \tefs in depth and showed that, unlike previous asynchronous hyperlogics, variants of our logic can express some of these properties (such as synchronicity and fairness) themselves. 
Table~\ref{tbl:overview} summarises the complexity results for the model checking problem.  
We have shown that the model checking problem is highly undecidable already for $\exteamctl$ with Boolean disjunctions.
However, we also showed that $\teamctls(\bor,\bneg, \NE, \flatop, \dep{}, \subseteq)$ can be translated to AABA over fixed-size teams, providing a general approach to defining restricted asynchronous semantics and obtaining decidability for the path checking problem as well as the finite model checking and satisfiability problems.

A possible direction for future work would be to study further restricted classes of \tefs beyond $k$-synchronicity and $k$-context-boundedness for which the corresponding emptiness problem for AABA is decidable, since this would lead not only to new decidable semantics for our logic but also for other logics such as H$_{\mu}$ \cite{GutsfeldMO21}.
Additionally, it would be interesting to look into different techniques to solve variants of the model checking and satisfiability problems in which the number of traces is not fixed.

\begin{acks}
The first two authors were partially funded by the DFG grant MU 1508/3-1. 
The third author was partially funded by the DFG grants ME 4279/1-2 and ME 4279/3-1. 
The fourth author was partially funded by the DFG grant VI 1045/1-1.
\end{acks}

\balance
\bibliographystyle{ACM-Reference-Format}
\bibliography{biblio,main}

\end{document}